\documentclass[superscriptaddress,amsmath,A4,pra,twocolumn,showpacs]{revtex4} 
\usepackage{amssymb}
\usepackage{amsthm}
\usepackage{graphicx}
\usepackage{color}
\usepackage{mathcomp}

\usepackage{amsmath}

\usepackage{xypic}
\newcommand{\Hs}{{\cal H}}
\def\tr{{\rm tr}}

\def\Supp{\mathsf{Supp}}
\def\<{\langle}
\def\>{\rangle}
\def\ket#1{|#1\>}
\def\bra#1{\<#1|}

\def\map#1{\mathcal{#1}}
\def\spc#1{\mathcal{#1}}
\def\set#1{\mathsf{#1}}
\def\Tr{\mathrm{tr}}
\def\st#1{\boldsymbol{#1}}

\newtheorem{Def}{Definition}

\newtheorem{proposition}{Proposition}
\newtheorem{corollary}{Corollary}
\newtheorem{theorem}{Theorem}

\newcommand{\tA}{{\mathcal{A}}} 
\newcommand{\tB}{{\mathcal{B}}} 

\newcommand{\Rt}{R_{\mathrm{t}}}   
\newcommand{\Rs}{R_{\mathrm{s}}}   
\newcommand{\Rm}{R_{\mathrm{m}}}   
\newcommand{\RC}{{R}}   

\begin{document}

\title{Incompatible measurements on quantum causal networks}
\author{Michal Sedl\'ak}
\affiliation{Department of Optics, Palack\'{y} University, 17. listopadu 1192/12, CZ-771 46 Olomouc, Czech Republic}
\affiliation{RCQI, Institute of Physics, Slovak Academy of Sciences, D\'ubravsk\'a cesta 9, 84511 Bratislava, Slovakia}

\author{Daniel Reitzner}
\affiliation{RCQI, Institute of Physics, Slovak Academy of Sciences, D\'ubravsk\'a cesta 9, 84511 Bratislava, Slovakia}

\author{Giulio Chiribella}
\affiliation{Department of Computer Science, The University of Hong Kong, Pokfulam Road, Hong Kong}
\author{M\'ario Ziman}
\affiliation{RCQI, Institute of Physics, Slovak Academy of Sciences, D\'ubravsk\'a cesta 9, 84511 Bratislava, Slovakia}
\affiliation{Faculty of Informatics,~Masaryk University,~Botanick\'a 68a,~60200 Brno,~Czech Republic}

\pacs{03.65.Ta, 03.67.-a, 03.65.-w, 03.65.Aa}

\begin{abstract}
The existence of incompatible measurements, epitomized by Heisenberg's uncertainty principle,  is one of the distinctive features of quantum theory.   So far, quantum incompatibility has been studied  for measurements that test the preparation of physical systems.  Here we extend the notion   to measurements that test dynamical processes, possibly consisting of multiple time steps.
Such measurements  are known as testers and are
 implemented by interacting with the tested process through a sequence of state preparations, interactions, and measurements.
Our first result is a characterization of the incompatibility of quantum testers, for which we provide  necessary and sufficient conditions.
  Then,  we propose a quantitative measure  of  incompatibility. We call this measure the robustness of incompatibility and define it as  the minimum amount of noise that has to be added to a set of testers in order   to make them  compatible. We show that  (i) the robustness is lower bounded by the distinguishability of the sequence of interactions used by the tester and (ii)  maximum robustness is attained when the interactions are perfectly distinguishable.
  The general results are illustrated in the concrete example of binary testers probing the time-evolution of a single-photon polarization.
  \end{abstract}

\maketitle

\section{Introduction}

Quantum theory   challenges our intuition in many ways,
a  prominent example being the existence of  incompatible measurements \cite{heisenberg}. 
Observable quantities, such as the position and the velocity of a particle, can be incompatible, in the sense that it is impossible to measure them in a single experiment unless   some  amount of noise  is added \cite{arturs-kelly}.
The existence of incompatible measurements is at the root of the wave-particle duality \cite{interference,coles} and, more generally, of many  striking distinctions  between quantum and classical physics.

For example,  the existence of incompatible measurements implies the no cloning theorem \cite{wootters,dieks}:  by contradiction, if we could make two perfect copies of an arbitrary quantum state, we could perform one measurement on one copy and another measurement on the other, so that no pair of measurements would be incompatible.  This argument applies also to approximate universal cloning \cite{BHcloner,gisinmassar,werner}, whose optimal performance is limited by the incompatibility of quantum measurements  \cite{hofmann}.   Moreover,  if all quantum measurements were compatible, then quantum states could be represented as probability distributions over a classical phase space, whose points would be labeled by outcomes of all possible quantum measurements.         For composite systems, such a classical description would prevent the violation of Bell inequalities \cite{wolf}, inhibiting important applications such as device-independent  cryptography \cite{hardy-barrett-kent,acin}.

More recently, the existence of incompatible measurements has been recognized as equivalent to the existence of Einstein-Podolski-Rosen (EPR) steering, a weaker form of non-locality whereby the choice of  measurement on one system determines the ensemble decomposition of the state of another system.  In short, the argument \cite{brunner,guehne} is as follows:  if all quantum measurements were compatible, then one could explain the phenomenon of steering in terms of a \emph{hidden state model} \cite{wiseman},  wherein  the state of the steered system  is defined before the measurement.  Vice-versa, if every instance of EPR steering could be  explained by a hidden state model, then all quantum measurements should be compatible.   Based on this argument, one can  establish a quantitative  connection between incompatibility and EPR steering, which has been explored extensively in Refs.~\cite{handchen,skr01,skr02,pijany}.

Due to its fundamental implications, quantum incompatibility  has been the object of intense research   \cite{ludwig,buschspin,lahti,pulmannova,qubit01,qubit02,qubit03,projectionalgebra,kunjwal,robustjukka} (see Ref. \cite{review} for a recent review).
So far, all investigations have focused on the standard scenario  where the goal of measurements is to
test  properties of the system's preparation   \cite{ludwig}.
However, one can consider more general scenarios, where the goal of the measurement is to test a property of a dynamical process \cite{opnorm,ziman,aharonov,dariano}. For example, imagine  that we are given an optical device that  transforms the polarization of photons in some unknown way.   To gain some knowledge of the device,  we can ask how well it preserves the vertical polarization.  This property can be tested by  preparing a vertically polarized photon, sending it through the device, and finally performing a polarization measurement with a vertically aligned polarizer.  Similarly, we may ask how well the device preserves the horizontal, or the diagonal polarization (defined as the polarization aligned by +45 degrees with respect to vertical polarization), or whether the device transforms the vertical polarization into the diagonal one. All these questions correspond to different experimental setups, that one can use to test different properties of the unknown process. These properties are often \emph{complementary} \cite{Bohr}  in the sense that they cannot be tested in a single experiment --- i.e.~they correspond to \emph{incompatible measurement setups}.

The  aim of this paper is to provide a precise characterization of what setups are incompatible when we  test  quantum processes.   Our prime  motivation  is fundamental: we want to explore the new forms of complementarity  arising from the study of quantum dynamics. On a more practical side, we expect that our generalized notion of quantum incompatibility will have new applications in quantum information, in the same way as the early studies of quantum complementarity and incompatibility  led to the discovery of new quantum protocols.   In our investigation   we will start from the simplest case of processes  that consist of a single time step.      To test such processes, we consider setups that consist of  (i) preparing the input of the process in a known state (possibly a joint state of the input with an ancillary system), (ii)  letting the state evolve through the process, and (iii) performing a measurement on the output, as in Figure \ref{fig:testerscheme}.

\begin{figure}
\begin{center}
\includegraphics{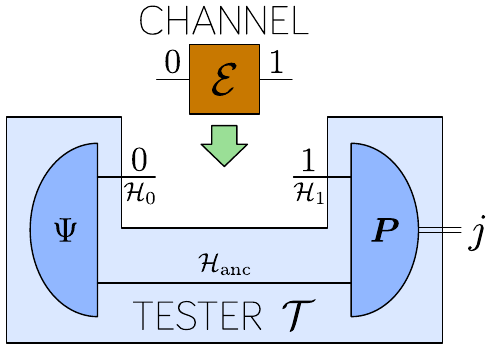}
\end{center}
\caption{\label{fig:testerscheme}
Pictorial representation of a setup testing a property of a
quantum process $\mathcal E$.  The input and output of the  process are  labeled as $0$ and $1$,   respectively.  The setup consists of the preparation of the input (and possibly an ancilla) into state $\Psi$  and in the execution of a measurement (POVM) $\st P$ on the output.
Any such setup can be described by a \emph{quantum tester}, a suitable generalization of the notion of positive operator-valued measure.
}
\end{figure}

 A measurement setup as in  Figure \ref{fig:testerscheme}  can be represented in a compact way using the notion of  \emph{quantum tester} \cite{opnorm,ziman,gutwat}, a generalization of the notion of \emph{positive operator-valued measure  (POVM)}  \cite{holevobook,qbook}.    More specifically, a tester is a collection of operators that can be used to compute the outcome probabilities for the setup under consideration.
It is important to stress that, like the notion of POVM,  the notion of tester involves a certain level of abstraction: since  the tester describes only the probabilities   of the outcomes,  different experimental implementations giving rise to the same statistics are identified. For example, suppose that we want to know whether a process is \emph{unital}, i.e.~whether it preserves the maximally mixed state.  A natural way to test unitality  is to prepare the maximally mixed state and to perform a state tomography on the output. However, the maximally mixed state can be prepared in many different ways:  for example, one could set up a stochastic mechanism that, with equal probabilities, prepares a photon with vertical or horizontal polarization.  Alternatively, the mechanism  could prepare a photon with diagonal (+45 degrees) and antidiagonal (-45 degrees) polarizations.  Or one could prepare two photons in a maximally entangled state, so that the reduced state of each photon is maximally mixed.  Despite being physically different, all these procedures will eventually lead to the same statistics, and, therefore, to the same tester.

The statistical  point of view  will be  crucial  for our notion of incompatibility.  We will regard two testers  as compatible if their outcomes can be generated in a single experiment and the corresponding  probability distribution   has marginals  coinciding with the probability distributions  predicted  by the original testers.  This is a purely information-theoretic  notion of compatibility:  it states that the statistics of two testers can be merged into the statistic of a third tester.  It is worth stressing, however, that the physical implementation of the third tester may   be very different from the physical implementations of the original testers.

A key point of our work will be to identify the sources of incompatibility that affect the tests of  dynamical processes.    A well-known source of incompatibility is the incompatibility of the measurements performed at the output: a tester that prepares a vertically polarized photon and measures  the output with a vertically oriented polarizer is incompatible with a tester that prepares a vertically polarized photon and measures the output with a diagonally oriented polarizer.  However, in the case of processes there is another source of incompatibility, namely the incompatibility of the inputs:  for example, one can test the action of a process on vertically polarized photons, or one can test it on horizontally polarized photons, but there is no joint setup that performs both tests at the same time.  Note that   taking a superposition of horizontal and vertical polarizations would  not work,  because the action of the process on the superposition does not give enough information about the action of the
process on the individual states that are being superposed.

One of the first results in our paper is a  necessary and sufficient condition  for the statistical compatibility of two (or more) testers.
Afterwards, we provide a quantitative measure of incompatibility, based on the amount of noise needed to make two (or more) testers compatible.    This notion, called \emph{robustness of incompatibility}, will allow us to  give an interesting lower bound, where the amount of incompatibility of  two setups  is lower bounded by the distinguishability of the input states used to probe the unknown process.   As a result of this bound,  we find that only setups with the same local input states can be compatible.
A complete analysis of the compatibility conditions is presented in the case of two-qubit, two-outcome testers.
All our results can be generalized to the case of processes consisting of multiple time steps \cite{opnorm,gutwat}, each step transforming an input into an output.
Such multi-time processes can be tested by preparing an input for the first step and applying a sequence of operations, as in Figure \ref{fig:Ncombsimple}.

\begin{figure*}
\begin{center}
\includegraphics{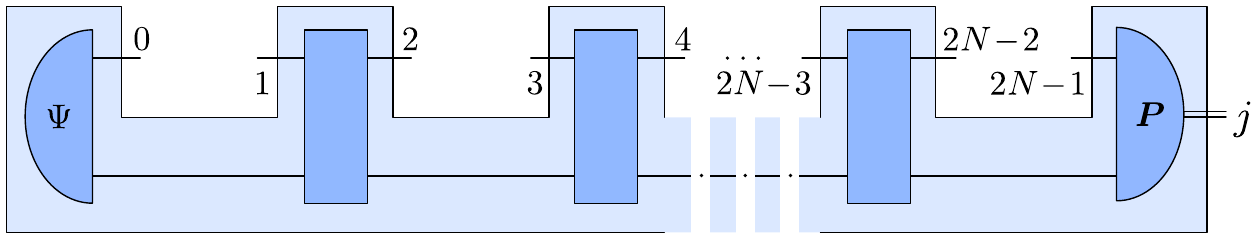}
\end{center}
\caption{\label{fig:Ncombsimple}
Pictorial representation of a measurement setup testing a property of multi-time
quantum processes.
}
\end{figure*}

The paper is structured as follows:
In Section \ref{sec:sec2} we introduce the mathematical framework
of process POVMs suited for the analysis of tester
incompatibility questions. In Section \ref{sec:incompatibility} we define incompatibility of testers.
The Section \ref{sec:sec3} introduces a measure
of incompatibility of testers that is evaluated in Sections \ref{sec:incompfsts} and \ref{sec:incompfrompbs} in cases
when the incompatibility is rooted in the incompatibility of the input states and final
measurements, respectively. In Section \ref{sec:sec5} we investigate in details the
incompatibility of two-outcome testers, especially, we focus
on factorized qubit case. In Section \ref{sec:sec6} we generalize the
incompatibility consideration for general
quantum networks and
we point out that the introduced incompatibility measure
is bounded from below by success probability characterizing the
minimum-error discrimination of corresponding quantum devices.
In Section \ref{sec:sec7} we summarize our findings. Technical results are gathered in the Appendix.

\section{Background on quantum testers}
\label{sec:sec2}

\subsection{Testing quantum processes}

Quantum testers \cite{opnorm,ziman,gutwat} provide a compact way of representing experimental setups designed to test unknown quantum processes.   Let us start from the simplest case, where the tested process consists of a single time step.
A setup testing  such processes consists of
\begin{enumerate}
\item the  joint  preparation of an  input system  and an ancilla,
\item  the application of the tested process on the input,  and
\item  the execution of a joint measurement  on the output  and the ancilla, as in Figure~\ref{fig:testerscheme}.
\end{enumerate}
In the following we will label the input, output, and ancilla system as $0$, $1$, and $\rm anc$, respectively.   We denote by $\spc H_0$,  $\spc H_1$, and $\spc H_{\rm anc}$ the corresponding Hilbert spaces and by $d_0$,  $d_1$, and $d_{\rm anc}$, the corresponding dimensions, respectively. Moreover, we denote by $\set S$ the set of possible outcomes of the final measurement.

Mathematically,  the  above setup is specified by a triple $\map T  =(\Hs_{\rm anc},   \Psi,  {\st P}  )$, where
\begin{enumerate}
\item $\Hs_{\rm anc}$ is the  Hilbert space of the ancilla used in the experiment,
\item  $\Psi$ is a density operator, acting on the tensor product Hilbert space $\Hs_0\otimes\Hs_{\rm anc}$ and representing a joint preparation of the input system and the ancilla, and
\item  ${\st P} = \{P_j,j \in \set S\}$  is a positive operator-valued measure (POVM) on  $\Hs_1\otimes\Hs_{\rm anc}$, representing a joint measurement on the output system and the ancilla.
\end{enumerate}
The tested process is described by a completely positive trace-nonincreasing linear map  $\map E$, transforming  operators on the input Hilbert space $\Hs_0$ into operators on the output Hilbert space $\Hs_1$.   For \emph{deterministic processes} (also known as \emph{quantum channels}) the map $\map E$ is trace-preserving, see e.g.~\cite{qbook}.

When the  process $\map E$ is tested with the setup $\map T$, the probability that the measurement produces the outcome $j$ is given by
\begin{align}\label{probabilities}
 p_j(\map T, \mathcal{E})=\tr{[  P_j  \,  (\mathcal{E}\otimes\mathcal{I}_{\rm anc})(\Psi) ]},
\end{align}
 where $\mathcal{I}_{\rm anc}$ is the identity mapping on the ancilla.
 The probabilities defined in this way are non-negative and sum up to one if the tested process is deterministic.  A remarkable property of quantum theory is that, under minimal requirements, every admissible map sending quantum processes to probability distributions can be physically implemented via some 
 setup $\map T$, meaning that one can always find an ancillary system, an input state, and a measurement that give rise to the desired mapping $\map E  \mapsto  p_j  (\map E)$  \cite{ziman,supermaps,dariano}.

The probabilities in Eq.~(\ref{probabilities}) can be written down in a compact way using the Choi isomorphism \cite{Choi}, whereby the process $  \map E$ is represented by the positive   (semidefinite) operator $E$ defined by
\begin{align}\label{Choi}
E   : =   \left( \map E \otimes \map I \right)   (|\Omega\>\<\Omega|)\, ,
\end{align}
where $|\Omega \>  \in  \spc H_0 \otimes \spc H_0$ is the unnormalized maximally entangled state
\begin{align}\label{omega}
|\Omega\>  : =   \sum_{m=1}^{d_0} \,  |m\>|m\> \, ,
\end{align}
$\{\ket{m} \}_{m=1}^{d_0}$ being an orthonormal basis for $\Hs_0$.

In terms of the Choi operator,  the outcome probabilities
can be rewritten as \cite{opnorm,ziman,gutwat}
\begin{align}\label{testerprob}
p_j (\map T,  \map E)   =    \Tr \left[   T_j  \,   E \right] \, ,
\end{align}
where $T_j$ is the operator on $\spc H_1\otimes \spc H_0$ defined by \cite{opnorm}
\begin{align}\label{Tj}
T_j:   =  \Tr_{\rm anc} [     (  P_j  \otimes  I_0 )  \,   (   I_1 \otimes   {\tt SWAP} \, \Psi^{T_0}  \, {\tt SWAP}  \, )].
\end{align}
Here $\Psi^{T_0}$ denotes  the partial transpose  of $\Psi$ on the Hilbert space $\spc H_0$, and $\tt SWAP$ is the unitary operator that swaps the Hilbert spaces $\spc H_0$ and $\spc H_{\rm anc}$  in order to have them consistently ordered.

It is easy to see that the operators  $\{  T_j  \, , j\in\set S\}$     satisfy the  conditions
\begin{subequations}
\begin{align}
\label{positivity}   {\rm positivity:} \quad   &  T_j   \ge 0  \, , \quad \forall j\in \set S     \\
\label{normalization}  {\rm normalization:} \quad & \sum_{j\in\set S} \,  T_j     =  I_1\otimes \rho  \, ,
\end{align}
\end{subequations}
where $\rho$ is a density operator on $\spc H_0$.    Physically, equations (\ref{positivity}) and (\ref{normalization})  guarantee the positivity and normalization of the outcome probabilities.

The above observations lead to the definition  of quantum tester:
\begin{Def}[\cite{opnorm}]
Let ${\st T}  =  \{  T_j,j\in  \set S\}$ be a collection of operators on $\spc H_1 \otimes \spc H_0$.  We say that $\st T$ is a \emph{quantum tester} if it satisfies the  conditions  (\ref{positivity}) and (\ref{normalization}), for some suitable density operator $\rho$ on $\spc H_0$.  We call the operator $\rho$ the \emph{normalization state} of the tester $\st T$.
\end{Def}
Quantum testers have also been called  \emph{process POVMs} in Ref.~\cite{ziman} and  \emph{measuring co-strategies} in Ref.~\cite{gutwat}.

When there is no ambiguity, we will omit the explicit specification of the outcome set.  For example, we will write $\st T  =  \{  T_j\}$ instead of $\st T  =  \{  T_j\, ,  j\in\set S\}$ and
\[  \sum_j   T_j   \quad {\rm instead~of}\quad     \sum_{j\in\set S}   T_j   \, . \]

\subsection{Physical implementation of quantum testers}
We have seen that every  experimental setup   testing  quantum processes can be   described by a tester.   The converse is also  true:  for every tester, one can find a setup that generates the corresponding statistics.

  \begin{Def}
We say that a setup   $ \map T  =  (\spc H_{\rm anc},  \Psi,   {\st P} )$ is a \emph{physical implementation} of the tester $\st T$ if it satisfies
the condition
\begin{align}
p_j( \map T,  \map E)  =  \Tr  \left[  T_j \,  E  \right]  \,  ,
\end{align}
for every outcome $j $ and for every process $\map E$.
\end{Def}

A canonical way to construct physical implementations is provided by the following:
\begin{proposition}[\cite{opnorm,ziman}]\label{prop:implementation}
For a given tester  ${\st T}  = \{  T_j  \}$,  Let
\begin{enumerate}
\item  $\rho$ be the normalization state in  Eq.~(\ref{normalization}),
\item $\spc H_\rho$ be the support of $\rho$,
\item     $|\Psi_\rho\>\in \spc H_0\otimes \spc H_\rho$  be  the unit vector defined by
\begin{subequations}
\begin{equation}
|\Psi_\rho\>    :=  \left(I_0  \otimes \rho^{\frac 12}\right)\,   |  \Omega  \>   ,
\end{equation}
 \item  $\Psi_\rho$ be the   density operator    $ \Psi_\rho  :  =  |\Psi_\rho\>\<\Psi_\rho| $, and
 \item  $\st P  = \{  P_j  \}$ be the POVM defined by
\begin{align}\label{canonicalPOVM}
P_j :  =       \left(I_0  \otimes \rho^{-\frac 12}\right)  \,  T_j \,    \left(I_0  \otimes \rho^{-\frac 12}\right)  \, ,
\end{align}
\end{subequations}
where $\rho^{-\frac 12}$ is the inverse of $\rho^{\frac 12}$ on its support.
\end{enumerate}
Then, the triple $\map T  =  ( \spc H_\rho,    \,  \Psi_\rho  \, , {\st P})$ is a physical implementation of the tester $\st T$.
\end{proposition}

\begin{Def}\label{def:canonical} We call  the POVM ${\st  P}$ defined in Eq.~(\ref{canonicalPOVM})  the \emph{canonical POVM} associated with the tester $\st T$.    The implementation defined in  Proposition  \ref{prop:implementation} will be called   the \emph{canonical implementation} of the tester $\st T$.
\end{Def}

Proposition  \ref{prop:implementation} tells us that every tester can be implemented with  an ancilla  of the size of the support of $\rho$, the normalization state associated with the tester.  In general, one can construct other implementations where the size of the ancilla is larger, or even smaller, as we will see later in an example. Nevertheless, all the physical implementations of a given tester must satisfy a common property, highlighted by the following proposition:

\begin{proposition}\label{prop:rhoT}
Let $\st T$ be a quantum tester and let $\map T  =  ( \spc H_{\rm anc},  \Psi, {\st P})$ be a physical implementation of $\st T$.   Then, the input state $\Psi$ must satisfy  the condition
\begin{align}
\Tr_{\rm anc}  [   \Psi ]   =   \rho^{T}  \, ,
\end{align}
where $\rho$ is  the normalization state  defined in Eq.~(\ref{normalization}) and the transposition is defined with respect to the basis  used in Eq.~(\ref{omega}).
\end{proposition}

A simple proof  can be found in Appendix \ref{app:proofrhoT}. In words,   proposition \ref{prop:rhoT} identifies the normalization state with (the transpose of) the local state on the input system.   Moreover, it implies  that all the physical implementations of the same tester must have the same marginal  state on the input system.  This property will play a crucial role in deciding the compatibility of testers.

\subsection{Ancilla-free testers}

 The simplest example of testers are those that can be implemented without  ancillas, as  in Figure \ref{fig:afreetesterscheme}.  \begin{figure}
\begin{center}
\includegraphics{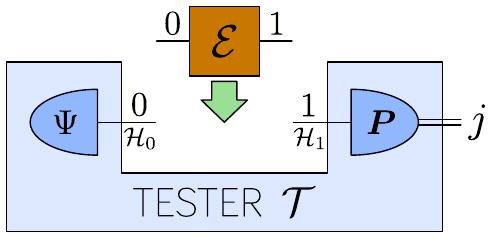}
\end{center}
\caption{\label{fig:afreetesterscheme}
 Diagrammatic representation of an ancilla-free quantum tester.}
\end{figure}
 Precisely, we adopt the following definition.
\begin{Def}
A tester $\st T$ is \emph{ancilla-free} if it admits an implementation $\map T   =  (\spc H_{\rm anc},  \Psi,  {\st P})$ where the ancilla Hilbert space is trivial, namely $\spc H_{\rm anc}   =  \mathbb C$.  When this is the case, we say that the implementation $\map T$ is \emph{ancilla-free}.
\end{Def}
   Ancilla-free testers have a very simple characterization.
   \begin{proposition}\label{prop:ancillafree}
   A tester ${\st T}  =\{  T_j\, ,  j  \in\set S \}$ is ancilla-free if and only if there exists a  POVM ${\st P}   =  \{  P_j\, , j\in\set S  \}$ and a density operator $\rho$ such that  one has
 \begin{align}\label{ancillafree}
 T_j  =   P_j\otimes \rho  \, ,  \qquad   \forall  j  \in  \set S  \, .
 \end{align}
Every ancilla-free implementation   $\map T   =   (\mathbb C,  \Psi,  {\st P} )$  has  $\Psi  =  \rho^T$.
\end{proposition}
A proof can be found in Appendix \ref{app:proofrhoT2}.
An example of ancilla-free tester is a tester designed to probe how an optical device preserves the vertical polarization of a single photon.
In this case, implementation consists in preparing photon in the vertical polarization state $\Psi  =  |V\>\<  V|$, feeding it in the input port of the device,  and  performing  the projective measurement $  {\st P}  =   \{  |V\>\<V|  \, ,  |H\>\<  H|  \}$ on the output.
The corresponding tester is
\begin{equation}
{\st  T}_V   =   \{   |V\>\<V|  \otimes    |V\>\<V|  \, ,      |H\>\<H| \otimes   |V\>\<V|  \}  .
\end{equation}
Another example is tester designed to probe  how the device preserves the horizontal polarization. In this case the tester is
\begin{equation}
{\st  T}_H   =   \{   |H\>\<H|  \otimes    |H\>\<H|  \, ,      |V\>\<V| \otimes   |H\>\<H|  \}  .
\end{equation}

Yet another example is a tester designed to probe the  unitality of a process.   For a process acting on the polarization of a photon, the test of unitality can be performed by preparing the input in the maximally mixed state $\Psi  =   I/2$ and by checking whether the output state $\map E   (   I/2)$ is still maximally mixed.  To  check, one can perform a tomographically complete POVM $\st P$ on the output and establish whether the statistics of the outcomes is compatible with the maximally mixed state.  For example, one can choose at random among three polarization measurements, performing the six-outcome POVM ${\st P}   =   \{  P_1,P_2,  \dots,  P_6\}$ with
\begin{align}\label{tomopovm}
\begin{array}{llll}
  P_1  & =   \frac  13  \,   |V\>\<  V  |  \, ,   \qquad     &    P_2  &  =   \frac  13  \,   |H\>\<  H  |  \\  \\
  P_3  &=       \frac  13  \,   |D\>\< D  | \, , & P_4  &=   \frac  13  \,   |A\>\< A  |\\  \\
  P_5 & =     \frac  13  \,   |R\>\<  R  | \, ,   &  P_6 &  = \frac  13  \,   |L\>\<  L  |  \, ,
  \end{array}
  \end{align}
   where $  |D\> \, ,  |A  \> \,,   |R\>  \,  ,  |L\>$ denote the states of diagonal, antidiagonal, right-handed,  and left-handed  polarization, respectively.     If the process $\map E$ is unital, the six possible outcomes should occur with equal probabilities. The above setup   corresponds to the tester  ${\st T}   =  \{  T_1, T_2, \dots,  T_6\}$  with
\begin{equation}
T_j  :=   P_j\otimes   \frac I2 \, , \qquad \forall  j\in  \{ 1,2, \dots,  6\} .
\end{equation}

Note that an alternative way to implement the same tester is to prepare  two-photon maximally entangled state
\[  |\Phi^+  \>    :  =  \frac { |  H\>  |H\>  +  |V\>|V\>}{\sqrt 2}\] and to perform POVM ${\st P'}  =  \{ P_1',   P_2'  , \dots  ,  P_6'\}$ with
\[  P_j'   :  =  P_j  \otimes I \,  .  \]
  This implementation is canonical (according to Definition \ref{def:canonical}),  but not ancilla-free:  it uses the second photon as ancilla. Besides the ancilla-free testers, other examples of testers are provided in Appendix \ref{app:examples}.

\section{Characterizing compatibility}
\label{sec:incompatibility}

Before analyzing the case of testers, we briefly define the notion of compatibility for POVMs. We refer the reader to Ref.~\cite{review} for a more in-depth presentation.

\subsection{Compatibility of POVMs}

POVMs provide statistical description of experiments designed to probe the preparation of quantum systems. Each POVM represents, so to speak, a different ``experimental question" that we can ask the system.
Given two  such  questions, it is important to know whether the system can answer both questions at the same time.  More precisely, it is important to know whether there exists a third experiment that produces the same outcomes as 
the two original experiments, with the same probability distributions.

Concretely, suppose that the system is prepared in the state $\rho$ and measured either with  the POVM  ${\st P}   =  \{  P_j \}$  or with the POVM $  {\st Q}   =\{   Q_k \}$.   In one case, the measurement will produce  the outcome  $j$ with probability
\begin{align}\label{p}
p_j  =    \Tr  [ P_j  \rho ] \, ,
\end{align}
while, in the other case the measurement will produce the outcome $k$ with probability
\begin{align}\label{q}
  \qquad  \,  q_k  =  \Tr [  Q_k  \rho]  \, .
\end{align}
Now, we would like to find an experiment that provides \emph{both} outcomes $j$ and $k$, with a joint probability distribution  $r_{jk}$ that  reproduces the statistics of the original measurements, when one takes the marginals:
 \begin{align}\label{marginal}     \sum_{k }    r_{jk}    = p_j  \qquad {\rm and}   \qquad   \sum_{j }   r_{jk}   =  q_k
 \end{align}
  for every $j$ and every $k$.   The statistics of the third experiment will  be determined by a (joint) POVM $ {\st R}   =  \{   R_{jk}  \}$, so that one has
  \begin{align}\label{r}
  r_{jk}   =  \Tr  [ R_{jk}  \rho]  \, ,
  \end{align}
  for every $j$ and $k$.

If the state $\rho$ is known, the condition (\ref{marginal}) can be trivially satisfied by choosing $R_{jk}=   p_j q_k  \,  I$, meaning that one can draw $j$ and $k$ at random according to the (known) probability distributions in Eqs.~(\ref{p}) and (\ref{q}).    The interesting scenario is when the state $\rho$ is not known to the experimenter.  Here    we require the condition (\ref{marginal}) to hold for every quantum state  $\rho$.
Under this requirement, Eq.~(\ref{marginal}) becomes equivalent to the condition
\begin{align}\label{compatiblePOVM}
   \sum_{k}  \,  R_{jk}    =  P_j      \qquad {\rm and}  \qquad
  \sum_{j} \, R_{jk}  =  Q_k
\end{align}
for every $j$ and every $k$.  The above discussion motivates the following definition.
\begin{Def}
Two POVMs $\st P=  \{  P_j\}$ and $\st Q=\{   Q_k\}$ are {\em compatible} if there exists a third POVM ${\st R}=  \{  R_{jk}\}$ such that Eq.~(\ref{compatiblePOVM}) is satisfied for every $j$ and $k$.   If no such POVM exists, we say that $\st P$ and $\st Q$ are \emph{incompatible}.
\end{Def}

An example of compatible POVMs is given by  \emph{commuting POVMs}, i.e.~POVMs  $\st  P$ and $\st Q$ satisfying the condition $[P_j,  Q_k]  = 0$ for every $j$ and $k$.   In this case,  one can define the joint POVM   ${\st R}=  \{  R_{jk}\}$ with operators
\begin{equation}
R_{jk}   :=    P_j Q_k ,
\end{equation}
whose positivity is guaranteed by the commutation of $P_j$ and $Q_k$.

On the other hand, there are examples of compatible POVMs that
are not commuting. A nice counterexample was introduced by Busch \cite{buschspin}, who considered  two unsharp measurements of horizontal/vertical and diagonal/antidiagonal polarizations, ${\st P}=\{P_1,I-P_1\}$ and
${\st Q}=\{Q_1,I-Q_1\}$, with
\begin{align}
P_1  &=\frac{1+p}{2} \,  |V\>\<  V|    + \frac{1-p}{2} \,  |H\>\<  H|  , \\
 Q_1 &=  \frac{1+q}{2} \,  |D\>\<  D|    + \frac{1-q}{2} \,  |A\>\<  A|       ,
\end{align}
 and showed that these are compatible whenever $p^2+q^2\leq 1$.  Clearly, the operators  $P_1$ and $Q_1$ do not commute, except in the trivial case when   $p$ or $q$ is zero.  In summary,  we obtained the following  observation.
 \begin{proposition}
For POVMs, commutativity implies compatibility,
but not vice-versa.
 \end{proposition}

\subsection{Compatibility of quantum testers}

In analogy with the POVM case, we can regard quantum testers as different ``experimental questions" that we can ask about  physical  processes.  The only difference is that now the questions we can ask are of a more dynamical nature --- essentially, they are questions about how  processes transform  different inputs.

Now,  we want to know whether two experimental questions can be answered at the same time.  Let us represent the two questions with two corresponding testers ${\st A}  = \{  A_j\} $ and ${\st B}  =  \{  B_k\}$, respectively, and let us represent the tested process $\mathcal E$ by its Choi operator $E$.
When the first, resp.~the second, question is  asked, the probability to obtain the outcome $j$, resp.~$k$  is given by
  \begin{align}\label{pqproc}   p_j  =    \Tr  [ A_j \, E  ]  \qquad\text{resp.}\qquad  q_k  =  \Tr [  B_k\, E ] \, ,
  \end{align}
 cf.~Eq.~(\ref{testerprob}).  Now, we are looking for an experimental setup that can answer both questions, by providing the outcomes $j$ and $k$ with the right probabilities.   In other words, we want to find a tester ${\st C}  =  \{ C_{jk}\}$, such that the probability distribution $\{ r_{jk} \}$ defined by
  \begin{align}
 r_{jk}   :=  \Tr  [ C_{jk}\,  E ]
 \end{align}
reproduces the statistics of the original testers, namely
\begin{align}     \sum_{k}    r_{jk}    = p_j  \qquad {\rm and}   \qquad   \sum_{j}   r_{jk}   =  q_k  \, .
 \end{align}

Requiring this condition to hold for every possible process (including both deterministic and non-deterministic processes) leads to the following definition.
\begin{Def}\label{def:marginals}
Two quantum testers ${\st A}=\{A_j\}$ and ${\st B}=\{B_k\}$ are
\emph{compatible} if there exists a (joint) tester ${\st C}=\{C_{jk}\}$
such that
\begin{align}
\label{compatibilitytester}
\sum_k C_{jk}=A_j  \qquad  {\rm and}  \qquad
\sum_j C_{jk}=B_k
\end{align}
for every $j$ and $k$.   If no such tester exists, we say that ${\st A}$ and ${\st B}$ are \emph{incompatible.}
\end{Def}

\subsection{Necessary and sufficient conditions for compatibility}

In this part we will analyze the immediate implications of our definition of compatibility.   First of all, it is easy to see that compatible testers must have the same normalization:
\begin{proposition}
\label{prop:necessary_normalization}
Let  $\st A$ and $\st B$ be two testers and let $\rho$ and $\sigma$ be the corresponding normalization states, as in Eq.~(\ref{normalization}).    If $\st A$ and $\st B$ are compatible, then one must have
\begin{align}\label{joint_conditions}
\rho  =  \sigma \, .
\end{align}
\end{proposition}
\begin{proof}
Let $ \st C$ be the joint tester that reproduces the statistics of $\st A$ and $\st B$.  Then, one has
\begin{align}
I_1\otimes\rho  &=  \sum_j A_j =  \sum_j  \left(  \sum_k  C_{jk}\right)\notag\\
    & =\sum_k \left( \sum_j C_{jk}\right) = \sum_k B_k =I_1\otimes\sigma.
\end{align}
Hence, one must have $\rho  =  \sigma$.
\end{proof}

As an application,  Proposition \ref{prop:necessary_normalization} shows that the tester
\begin{align}\label{TV}  {\st  T}_V   =   \{   |V\>\<V|  \otimes    |V\>\<V|  \, ,      |H\>\<H| \otimes   |V\>\<V|  \}  \, ,
\end{align}
designed to probe the preservation of the vertical polarization, is incompatible with the tester
\begin{align}\label{TH}  {\st  T}_H   =   \{   |H\>\<H|  \otimes    |H\>\<H|  \, ,      |V\>\<V| \otimes   |H\>\<H|  \}  \, ,\end{align}
designed to probe the preservation of the horizontal polarization.    Indeed, ${\st T}_V$ has the normalization state $\rho =  |V\>\<  V|$, while ${\st T}_H$ has the normalization state $\sigma  =  |H\>\<  H|$.   Interestingly,  all the operators in the testers ${\st T}_V$ and ${\st T}_H$ commute.    Summarizing, we have the following observation.
\begin{proposition}
For testers, commutativity does not imply compatibility.
\end{proposition}

As the above counterexample shows, the reason why two commuting testers may fail to be compatible is that their normalization states do not coincide.  Interestingly, once this obstacle is removed, commuting testers become compatible:
\begin{proposition}
Two commuting testers $\st A$ and $\st B$ are compatible if and only if they have the same normalization states.
\end{proposition}

\begin{proof}  The ``only if" part follows immediately from proposition \ref{prop:necessary_normalization}.  Let us show the ``if" part.    Denote by $\rho$ the normalization state.  Since the testers commute, the normalization state  commutes with all the operators:   indeed, one has
\begin{align}
[I_1\otimes \rho  \, ,   B_k]  &  =   \sum_{j} \,   [ A_j ,  B_k]   =0, \\
[A_j,  I_1\otimes \rho]  &  =   \sum_{k} \,   [ A_j ,  B_k]   =0  ,
\end{align}
for every $j$ and $k$.
As a consequence, also the inverse $(I_1\otimes \rho)^{-1}$ (on the support of $\rho$)  commutes with $A_j$ and $B_k$ for every $j$ and $k$.  Using this fact, we  define the tester  ${\st C}=  \{ C_{jk}\}$ with operators
\begin{equation}
C_{jk}   =  A_j  B_k    (I_1 \otimes \rho)^{-1},
\end{equation}
whose positivity  follows from the commutation of $A_j$, $B_k$ and $(I_1\otimes \rho)^{-1}$. It is immediate to verify the normalization condition (\ref{normalization}) and the compatibility conditions (\ref{compatibilitytester}).
\end{proof}

More generally, we now provide a necessary and sufficient condition for the compatibility of two arbitrary (possibly non-commuting) testers.  The condition is expressed in terms of the canonical implementation of the testers introduced in Proposition \ref{prop:implementation}.

\begin{theorem}\label{prop:norm}
Two quantum testers $\st A$ and $\st B$ are compatible if and only if
\begin{enumerate}
\item the normalization states  associated with $\st A$ and $\st B$ in Eq.~(\ref{normalization}) coincide, and
\item the canonical POVMs associated with $\st A$ and $\st B$  are compatible.
\end{enumerate}
\end{theorem}

The proof can be found in Appendix \ref{app:proofnorm}.  The physical meaning of Theorem \ref{prop:norm} is that two testers   $\st A$ and $\st B$ are compatible if and only if they can be implemented with two experimental  setups  $\map A  =  (  \spc H_{\rm anc},  \Psi,  {\st P})$ and $\map B  =  (  \spc H_{\rm anc}  \, ,  \Psi \, ,  {\st Q})$ using the same ancilla, preparing the same input state, and measuring compatible POVMs   $\st P$ and $\st Q$.
  
Note that thanks to  Theorem \ref{prop:norm}, we do not need to search over all possible physical implementations of the two testers. All the information regarding the compatibility of the testers can be read out from their canonical implementation.

\section{Quantifying incompatibility}
\label{sec:sec3}

The incompatibility of POVMs is a resource  for various quantum information tasks, such as  steering \cite{brunner,guehne} or device independent cryptography \cite{hardy-barrett-kent,acin}. It plays a strong role also in approximate cloning \cite{hofmann}. 
   It is then natural to expect that also  the incompatibility of testers  will serve as a resource for information processing, at a higher level where the information is encoded into processes, rather than states.   From this point of view, it is natural to look for suitable measures that quantify the amount of incompatibility of two or more testers.  In the following we will provide one such measure, which we call the robustness of incompatibility.    Our constriction is inspired  by a convex method previously used  for POVMs \cite{busch,robustjukka,haapasalo}.

\subsection{Convexity of the set of testers}

Testers with the same set of outcomes form a convex set.  Precisely, we have the following

\begin{Def}\label{def:convextester}
Let  ${\st A}  =  \{  A_j  \, ,  j\in\set S  \}$ and ${\st B}  =  \{  B_j \, ,  j\in\set S\}$
 be two sets of operators, acting on the same Hilbert space and indexed by the same index set $\set S$.   The convex combination of   $\st A$ and $\st B$ with weight  $\lambda  \in  [0,1]$ is the set of operators $\st C  =  \{  C_j \, , j\in\set S\}$ defined by
\begin{equation}\label{eq:convextester}
C_j  :  =  (1-\lambda) \, A_j   +   \lambda\,  B_j \, .
\end{equation}
We will denote the convex combination as
\begin{align}\label{cc}
{\st C}:= (1-\lambda) \, {\st A} +  \lambda\,  {\st B} \,   .  \end{align}
\end{Def}
Denoting by $  \set T   (  \spc H_0,  \spc H_1,\set S)$ the set of all testers with input space $\spc H_0$, output space $\spc H_1$, and outcome space $\set S$, we have the following
\begin{proposition}
The set  $ \set T   (  \spc H_0,  \spc H_1,\set S) $ is convex.
\end{proposition}
\begin{proof}
It is easy to see that  the set of operators $\st C$ defined in Eq.~(\ref{cc}) is indeed a tester:  the operators $C_j$ are obviously positive and one has the normalization condition
$\sum_j  \,  C_j   =    I_1  \otimes \tau$ with  $\tau   = (1- \lambda) \rho   +   \lambda  \,  \sigma$ ,
where $\rho$ and $\sigma$ are the normalization states of  $\st A$ and $\st B$, respectively. Finally, convexity of the set of density matrices  implies that the operator $\tau$ is a density matrix.  Hence, the set of operators $\st C$ is properly normalized.
\end{proof}

\subsection{The robustness of  incompatibility}

Operationally, convex combinations of testers corresponds to the randomization of different experimental setups.    In particular, we can consider the case where an ideal setup, designed to measure a quantity of interest,   is randomly mixed with another setup, regarded as noise.

Now, suppose that the ideal tester $\st A$ is mixed with the noise  ${\st N}^{(\st A)}$ and that the ideal tester  $\st B$ is mixed with the noise ${\st N}^{(\st B)}$.  By adding enough noise, we can make the resulting testers compatible.  This idea motivates the following

\begin{Def}\label{lambdacomp}
Two testers $\st A$ and $\st B$ are  \emph{$\lambda$-compatible}  if there exist two testers  ${\st N}^{(\st A)}$ and $\st N^{(\st B)}$ such that the randomized testers
$  (1-\lambda) {\st A}  +  \lambda   {\st N}^{(\st A)}$ and $  (1-\lambda) {\st B}  +  \lambda   {\st N}^{(\st B)}$
 are compatible.
\end{Def}
In this setting the probability $\lambda$ can be interpreted as a measure of the level of added noise.  Intuitively, the amount of noise  needed to break the incompatibility of two testers can be regarded a measure of the degree of incompatibility: the higher the noise, the higher the incompatibility.  Based on this idea we define a quantitative measure of incompatibility, in terms of the minimum amount of added noise required to make two testers compatible:

\begin{Def}\label{robustness_incompatibility}
The {\em  robustness of  incompatibility}  $\Rt({\st A},{\st B})$ is the minimum $\lambda$ such that the testers $\st A$ and $\st B$ are
$\lambda$-compatible.
\end{Def}

In the rest of the paper we will analyze  the properties of the
robustness of incompatibility.
\subsection{Why different noises?}

Our definition of the robustness of incompatibility is similar to the definition of robustness used in the  literature on   POVMs \cite{robustjukka, haapasalo, heinosaari}.  Except in one detail:  for POVMs,  one usually assumes that the added noises coincide, i.e. ${\st N}^{(\st A)}=\st N^{(\st B)}$.  We do not make this assumption here, because assuming equal noises in the case of testers would result in a notion of robustness with undesired physical properties.

This point is clarified by the following proposition, proven in Appendix \ref{app:nostrong}.
\begin{proposition}
\label{prop:common_noise}
Let  $\st A$ and $\st B$  be two testers with distinct  normalization states and let $\st N$ be an arbitrary tester, representing the noise.  Then, the testers  $  (1-\lambda) {\st A}  +  \lambda   {\st N}$ and $  (1-\lambda) {\st B}  +  \lambda   {\st N}$ are compatible only  if $\lambda =  1$.
\end{proposition}

\medskip
In other words, if two testers have different normalization states, then adding the same noise on both will not make them compatible, except in the trivial case when the noise completely replaces the original testers.   As a result, even two testers that are arbitrarily close to each other would be maximally incompatible, just because they have slightly different normalization states.  A tester could be maximally incompatible with a noisy version of itself, even if the noise is arbitrarily small. For these reasons, we regard the definition of robustness with equal added noises as not physically interesting in the case of testers.

\subsection{Bounds on the robustness of incompatibility}

The robustness of incompatibility satisfies some obvious bounds: first one has the lower bound  $\Rt({\st A},{\st B})\geq 0$.
   An upper bound  is given
by the following
\begin{proposition}\label{prop:maxlambda}
Every pair of testers $\st A$ and $\st B$ is $\lambda$-compatible
with $\lambda=1/2$.
\end{proposition}
\begin{proof}
It is enough to set  ${\st N}^{(\st A)}   :=  \st B$ and   ${\st N}^{(\st B)}   :=  \st A$. With this choice, the randomized testers  $1/2\, {\st A}  +  1/2\,   {\st N}^{(\st A)}$ and $  1/2\,  {\st B}  +  1/2\,   {\st N}^{(\st B)}$ coincide and therefore are trivially compatible.
\end{proof}

In summary, the robustness of incompatibility for a pair of testers has values in the interval
\begin{align}\label{bounds}
0\leq \Rt({\st A},{\st B})\leq \frac 12  \, .
\end{align}
The lower bound is attained if and only if the testers $\st A$ and $\st B$ are  compatible. In the following we will see that the upper bound is also attainable.

\section{Incompatibility due to the input states}
\label{sec:incompfsts}
The incompatibility between two testers  can have two different origins: it can come from the incompatibility of the input states, or from the incompatibility of the measurements on the output.   In this section we quantify  the contribution of the  input states.

\subsection{The robustness of state incompatibility}

We know that, in order to be compatible, two testers must have the same normalization state. When this is not the case, it is possible to give a lower bound on the incompatibility of the testers in terms of the normalization states.

\begin{proposition}\label{prop:normalization}
Let $\st A$ and $\st B$ be two testers and let $\rho$ and $\sigma$ be their normalization states, respectively. If $\st A$ and $\st B$ are $\lambda$-compatible, then there must  exist two states
$\widetilde{\rho}$ and $ \widetilde{\sigma}$  such that
\begin{align}
\label{eq:ncnorm}
(1-\lambda)\rho+\lambda \widetilde{\rho}=(1-\lambda)\sigma+\lambda \widetilde{\sigma}\,.
\end{align}
Hence, defining the \emph{robustness of state incompatibility}
\begin{align}
\Rs  (\rho, \sigma )  :  =  \min  \{  \lambda  ~|~   \exists  \widetilde \rho \, ,   \widetilde \sigma  \,  : {\rm  Eq.}~(\ref{eq:ncnorm})~{\rm holds}    \}
\end{align}
we have the lower bound
\begin{align}\label{lowerb}
 \Rt({\st A},{\st B})   \ge \Rs(\rho, \sigma )   \, .
\end{align}
\end{proposition}
\begin{proof}
Suppose that the testers $  (1-\lambda) {\st A}  +  \lambda   {\st N}^{(\st A)}$ and $  (1-\lambda) {\st B}  +  \lambda   {\st N}^{(\st B)}$ are compatible.  Then, they must have the same normalization states (see Proposition \ref{prop:necessary_normalization}).  On the other hand, the normalization states are  $(1-\lambda)\rho+\lambda \widetilde{\rho}$ and $(1-\lambda)\sigma+\lambda \widetilde{\sigma}$, where $\widetilde \rho$ and $\widetilde \sigma$ are the normalization states of ${\st N}^{(\st A)}$ and ${\st N}^{(\st B)}$, respectively.   Hence, Eq.~(\ref{eq:ncnorm}) should hold. Minimizing over $\lambda$ one finally gets the lower bound (\ref{lowerb}).
\end{proof}

\subsection{Maximally incompatible testers}

Using the lower bound (\ref{lowerb}), it is easy to construct quantum testers that achieve the maximum value of the robustness of incompatibility.    We call these testers maximally incompatible:
\begin{Def}
We say that two testers $\st A$ and $\st B$ are \emph{maximally incompatible} if  $\Rt({\st A},{\st B})  = 1/2$.
\end{Def}
An example of maximally incompatible tester is as follows:  consider two testers $\st A$ and $\st B$  that probe qubit processes and suppose that  $\st A$ and $\st B$     have  normalization states $\rho  =  |0\>\<0|$ and $\sigma  =  |1\>\<1|$, respectively.
   Then, Eq. ~(\ref{eq:ncnorm}) implies the inequality
    \begin{align}
\nonumber   \lambda   & \ge   \lambda    \,  \<  1|   \widetilde \rho  |1\>   \\
 \nonumber  &  =         \<  1| \,  \left[   (1-\lambda) \,  |0\>\<  0| +   \lambda \, \widetilde \rho \right] \,  |1\>\\
 \nonumber  &  =        \<  1| \,  \left[   (1-\lambda) \,  |1\>\<1|   +   \lambda \, \widetilde \sigma \right] \,  |1\>\\
\label{derivationbound}   & \ge  1-\lambda,
   \end{align}
which in turn implies $\lambda  \ge 1/2$.  Since the inequality must hold for every $\lambda$, the robustness of state incompatibility   should also satisfy the inequality $\Rs (\rho, \sigma )  \ge \frac 12$. Hence,  Eq.~(\ref{lowerb})  yields the lower bound
\begin{equation}
\Rt({\st A},{\st B})       \ge  \Rs  (\rho, \sigma )  \ge \frac 12  .
\end{equation}
Combining it with the upper bound (\ref{bounds}) we obtain the equality
\begin{equation}
\Rt({\st A},{\st B})      = \frac 12.
\end{equation}
Concrete example of  maximally incompatible qubit testers are the optical testers ${\st T}_V$  and ${\st T}_H$ designed to test the preservation of the vertical polarization and horizontal polarization, cf.~Eqs.~(\ref{TV}) and (\ref{TH}) for the explicit expression. Note that  the testers ${\st T}_V$  and ${\st T}_H$ are commuting, and are still, maximally incompatible.
For testers, not only commutativity does not imply compatibility, but also there exist commuting testers that are maximally incompatible!

\medskip

\noindent{\bf Remark (testers vs POVMs).}   The attainability of the upper bound $\Rt({\st A},{\st B})  =  1/2$  highlights an important  difference between the incompatibility of testers and  the incompatibility of POVMs. For POVMs,   it is known that the upper bound is not achievable for finite dimensional systems \cite{busch,robustjukka} 
(see Section \ref{sec:incompfrompbs} for qualitative details). In contrast,  the robustness of tester incompatibility can reach the maximum value  $\Rt({\st A},{\st B})  =  1/2$ even for two-dimensional systems.
For POVMs, this maximum value can only be attained for infinite dimensional
systems \cite{review,heinosaari}.

\subsection{Orthogonal testers are maximally incompatible}

Adapting the  inequalities used in  Eq.~(\ref{derivationbound}), one can see that  the robustness of incompatibility attains its maximum value $  \Rt({\st A},{\st B})  =  1/2$ whenever the testers $\st A$ and $\st B$  are orthogonal, in the following sense:
\begin{Def}
Two testers   $\st A$ and $\st B$  are orthogonal if one has $A_jB_k=  0$ for every pair of outcomes $j$ and $k$.
\end{Def}
It is easy to see that the orthogonality of two testers  is equivalent to the orthogonality of their normalization states.
\begin{proposition}
Let $\st A$ and $\st B$ be two testers and let $\rho$  and $\sigma$ their normalization states, respectively.
Then, the following are equivalent
\begin{enumerate}
\item $\st A$ and $\st B$ are orthogonal
\item $\rho$ and $\sigma$ are orthogonal, namely $  \rho \,  \sigma = 0$.
\end{enumerate}
\end{proposition}
\begin{proof}
By definition, one has $  I_1  \otimes \rho\,\sigma  =  \sum_{j,k}   A_j B_k$.  Hence, condition 1 implies condition 2.  On the other hand, taking the trace on both sides one has
\begin{equation}
d_1  \,\Tr  [ \rho  \, \sigma ]   =   \sum_{j,k}   \Tr [A_j B_k] \ge \Tr [A_j B_k]  \ge 0  \, ,
\end{equation}
having used the fact that $A_j$ and $B_k$ are positive.
Hence, the condition 2 implies $\Tr [A_jB_k]  =  0$, and, since $A_j$ and $B_k$ are positive, $A_jB_k= 0$.
\end{proof}
Note that orthogonal testers are commuting, because the condition $A_j  B_k= 0$ trivially implies $[A_j,  B_k]  =0$.

In summary, we have seen that  orthogonal  testers  are always maximally incompatible.   It is an open problem whether maximally incompatible testers must be orthogonal.

\subsection{The case of jointly diagonal testers}

In some cases, the robustness of tester incompatibility coincides with the robustness of state incompatibility, meaning that the lower bound (\ref{lowerb}) is attained.  One such case  involves pairs of  jointly diagonal testers, defined as follows.
\begin{Def} Two testers ${\st A}=\{A_j\}$ and ${\st B}=\{  B_k\}$ are \emph{jointly diagonal} if all the operators $A_j$ and $B_k$ are diagonal in the same basis.
\end{Def}
  Note that jointly diagonal testers are a strict subset of the set of commuting testers: while commuting testers satisfy the relation  $[A_j,  B_k]  =  0$ for every $j$ and $k$, jointly diagonal testers have also to satisfy the relations $[A_j, A_k]=0$ and $[B_j,  B_k]=0$, for every $j$ and $k$.    For jointly diagonal testers we have the following.

\begin{proposition}\label{prop:diagonal}
Let $\st A$ and $\st B$ be two jointly diagonal testers and let $\rho$ and $\sigma$ be the corresponding normalization states.
Then, one has
\begin{equation}
\Rt({\st A},{\st B})   =   \Rs  (\rho, \sigma ).
\end{equation}
\end{proposition}
The proof can be found in Appendix \ref{app:diagonal}.

An obvious example of jointly diagonal testers is the example of the testers ${\st T}_V$ and ${\st T}_H$, designed to test the preservation of the vertical and horizontal polarizations, respectively.
This kind of incompatibility originates in the mutually exclusive choices of inputs needed to test the properties represented by  ${\st T}_V$ and ${\st T}_H$.
Again, we stress that the incompatibility of inputs is sufficient to make two testers maximally incompatible.

\subsection{Computing the robustness of state incompatibility}

The robustness of state incompatibility  has a direct interpretation in terms of  distinguishability.  Specifically, we have the following
 \begin{proposition}\label{prop:discrimination}
For every  pair of density operators $\rho$ and $\sigma$, the robustness of state incompatibility is given by
\begin{align}
\label{eq:lambda_mingen}
 \Rs  (\rho, \sigma)=\frac{\|\rho-\sigma\|}{\|\rho-\sigma\|+2}\, \, ,
\end{align}
where $\|  \cdot \|  =  \Tr |\cdot |$ denotes the trace norm.
\end{proposition}
The above result can be concisely derived from the semidefinite programming approach to minimum error state discrimination, due to Yuen, Kennedy, and Lax \cite{Yuenkenlax}, combined with the operational interpretation of the trace distance, following from Helstrom's theorem \cite{helstromth}.   A more explicit derivation with a nice geometric interpretation can be found in Appendix \ref{app:discrimination}.

In the special case of two-dimensional systems (qubits), the robustness of state incompatibility has a simple expression in terms of the Bloch vectors of the states $\rho$ and $\sigma$, i.e.~the  vectors $\st r  =  (r_x,r_y,r_z)$ and $\st s  =  (s_x,s_y,s_z)$ in the expressions
\begin{equation}
\rho=\frac{1}{2}(I+{\st r}\cdot\boldsymbol{\sigma}) \quad \text{and}  \quad\sigma=\frac{1}{2}(I+{\st s}\cdot\boldsymbol{\sigma}) ,
\end{equation}
where $\boldsymbol{\sigma}  =  (\sigma_x,\sigma_y,\sigma_z)$
is the vector of the Pauli matrices.  Indeed, using the expression
\begin{equation}
\|  \rho-\sigma\|   =  \|    {\st r} -{\st s} \|    :=   \sqrt{  \sum_{i=x,y,z}   (  r_i  -  s_i )^2},
\end{equation}
one obtains the following corollary.
\begin{corollary}
For every  pair of qubit density operators $\rho$ and $\sigma$, the robustness of state incompatibility is given by
\begin{align}
\label{eq:lambda_min}
\Rs  (\rho,\sigma)=\frac{\|\bf{r}-\bf{s}\|}{\|  {\st r}-{\st s}\|+2}\,.
\end{align}
\end{corollary}

\section{Incompatibility due to the output measurements}
\label{sec:incompfrompbs}
In the previous section we quantified how the normalization states affect the incompatibility of two testers.  In this section we carry out a similar analysis for the  measurements.  We focus our attention on testers that have the same normalization states --- the rationale being that in such scenario the incompatibility arises purely from the measurements.

\subsection{Upper bound on the measurement-induced incompatibility}
Two POVMs $\st P$ and $\st Q$, representing two measurements on the same quantum system, are said to be \emph{$\lambda$-compatible} \cite{haapasalo} if there  exist two POVMs $ {\st J}^{(\st P)}$ and ${\st J}^{(\st Q)}$ such that the POVMs
\begin{equation}
(1-\lambda)  \,  {  \st P}  +  \lambda \,    {\st J}^{(\st P)}   \quad\text {and}  \quad (1-\lambda)  \,   { \st Q}  +  \lambda\,    {\st J}^{(\st Q)}
\end{equation}
are compatible.    In this context,  $ {\st J}^{(\st P)}$ and $ {\st J}^{(\st Q)}$ are regarded as introducing ``noise''  (or ``junk") in the statistics of the desired measurements $\st P$ and $\st Q$.
The robustness of incompatibility is then defined in the natural way.
\begin{Def}
\label{def:measurement_robustness}
The \emph{robustness of incompatibility} of two POVMs $\st P$ and $\st Q$,  denoted by $\Rm(  \st P,\st Q) $,  is the minimum $\lambda$ such that $\st P$ and $\st Q$  are $\lambda$-compatible.
\end{Def}

Now, when two testers have the same normalization state, one can upper bound their incompatibility in terms of the incompatibility of the canonical POVMs introduced  in definition (\ref{def:canonical}).  The upper bound is as follows:

\begin{proposition}
\label{prop:reltopovm}
Let $\st A$ and $\st B$ be two testers with the same normalization state $\rho$  and let $\st P$ and $\st Q$ be the canonical POVMs associated with $\st A$ and $\st B$, respectively. Then, one has
\begin{align}\label{measureinc}
\Rt({\st A},  {\st B})  \le   \Rm(  \st P,  \st Q ) \, .
\end{align}
\end{proposition}
\begin{proof}
The proof is  straightforward:  one way to make $\st A$ and $\st B$ compatible is to take their canonical implementations and add enough noise to the canonical POVMs in order to make them compatible. In other words, if $\st P$ and $\st Q$ are $\lambda$-compatible, then also  $\st A$ and $\st B$ are $\lambda$-compatible.  Taking the minimum over $\lambda$, one obtains the desired upper bound.
\end{proof}

\subsection{Reachability of maximal incompatibility}

As we mentioned earlier,  the robustness of incompatibility of two POVMs $\st P$ and $\st Q$ in finite dimensions is always smaller than $1/2$ \cite{review,heinosaari}. Physically, this can be seen by considering the process of optimal universal cloning \cite{BHcloner,gisinmassar,werner}, which outputs approximate copies of the original state,  mixed with a suitable amount of white noise.  By using the optimal cloner, one can make every two POVMs  compatible by measuring different copies, at the price of a noisy statistics \cite{hofmann}. For a quantum system of dimension $d$, the above procedure gives the bound
\begin{equation}
\Rm(\st P,\st Q)\leq\frac{1}{2}\left(1-\frac{1}{1+d}\right)\,.
\end{equation}
Consequently,  two testers with the same normalization states cannot  reach the maximum value of the robustness $ \Rt (\st A, \st B)= 1/2 $, unless the dimension of the output system is infinite. For systems of finite dimensions, it is natural to ask what the maximum amount of measurement-induced incompatibility is.    We conjecture that the maximum amount is achieved by setups that measure two mutually unbiased bases on the output; further discussion on this point can be found in Appendix \ref{app:fourier}.

\subsection{The case of testers with pure normalization state}

The upper bound $\Rt({\st A},  {\st B})  \le \Rm({\st P},  {\st Q})$ comes from using the canonical implementation of the testers $\st A$ and $\st B$.    In principle, however, the bound may not be saturated, because some non-canonical implementation may have ``more compatible" POVMs than the canonical implementation.   The bound \emph{is} saturated, however, in the case of testers with  pure normalization state:

\begin{proposition}
\label{prop:purenormcons}
Let $\st A$ and $\st B$ be two testers and let $\st P$ and $\st Q$ be their canonical POVMs. If the testers $\st A$ and $\st B$ have  the  same  pure normalization state,  then the equality
\begin{equation}
\Rt({\st A},  {\st B})  =  \Rm(  \st P,  \st Q )
\end{equation}
holds.
\end{proposition}

The proof can be found in Appendix \ref{app:purenormcons}.   Note that all testers with pure normalization state are ancilla-free, that is, they can be implemented by preparing the (transpose of) the normalization state, applying the tested process, and  measuring the output.   It is an open question whether  the equality $\Rt({\st A},  {\st B})  =  \Rm(  \st P,  \st Q )$ holds for all ancilla-free testers.

\section{Incompatibility of two-outcome testers}
\label{sec:sec5}

Here we consider the case of two-outcome testers ${\st A}  =  \{A_1,A_2\} $ and  ${\st B}  =  \{B_1,B_2\} $.   Even in this simple case,  we will see that the  incompatibility of testers is  a  subtler issue than the incompatibility of POVMs. We will first provide  general results, later moving to a concrete example  in the qubit case.

\subsection{Characterization of  compatibility}

In the case of twe-outcome testers, compatibility has a simple algebraic characterization:

\begin{proposition}\label{prop:SDP}
A pair of two-outcome quantum testers ${\st A}  =  \{A_1,A_2\} $ and  ${\st B}  =  \{B_1,B_2\} $
   are compatible if and only if
   \begin{enumerate}
   \item  $\st A$ and $\st B$ have the same normalization state, denoted by  $\rho$
   \item there exists a positive operator $C$ such that
\begin{align}
 \nonumber C &\leq A_1,   \\
  \nonumber  C  &  \leq  B_1,  \\
   I_1\otimes \rho  + C  &\ge A_1 +  B_1\, .
 \label{twooutcomecomp}
\end{align}
\end{enumerate}
\end{proposition}

\begin{proof}
By definition, $\st A$ and $\st B$ are compatible if and only if there exists a tester  ${\st C}$ such that
\begin{align}
A_1   &=    C_{11}  +  C_{12},\nonumber\\
B_1  &=  C_{11}   +   C_{21},\nonumber\\
A_2  &=  C_{21}  +  C_{22},\nonumber\\
B_2  &  =  C_{12}  + C_{22}  .
\end{align}
Setting $C:=C_{11}$, the above equations imply
 \begin{align}
 C_{12}&=A_1-C, \nonumber \\
 C_{21}&=B_1-C, \nonumber \\
 C_{22}&=I_1\otimes \rho-  C_{11}   -   C_{12}  -  C_{21}\nonumber\\
 &     =    I_1\otimes \rho    -A_1-B_1  +  C\, .
 \end{align}
The positivity of the above operators is equivalent to the conditions  in Eq.~(\ref{twooutcomecomp}).
\end{proof}

We stress  that the equality of the normalization states is essential for compatibility. Even in the extreme case $A_1=B_1$ the testers $\st A$ and $\st B$ may not be compatible, due to the fact that they have different normalization states.  In the following we will see that, when $A_1=  B_1$, the incompatibility of the testers $\st A$ and $\st B$ is equal to the robustness of normalization.

\subsection{The case of comparable testers}

Here we consider the case of \emph{comparable two-outcome testers}, defined as testers satisfying one of the relations
\begin{equation}
A_j  \le B_k    \qquad  {\rm or}  \qquad  B_k \le A_j
\end{equation}
for at least one pair of outcomes $(j,k)  \in  \{1,2\}\times \{1,2\}$.    As a special case, testers with $A_1  =B_1$ are comparable.

Unlike in the case of POVMs, where comparability implies  compatibility \cite{teikoalone},    comparable testers may not be compatible.
However, they have a remarkable property, expressed by the following proposition.

\begin{proposition}\label{prop:comparable}
For comparable two-outcome testers $\st A=\{A_1,A_2\}$ and $\st B=\{B_1,B_2\}$ with normalizations $\rho$ and $\sigma$ the robustness of incompatibility is equal to the robustness of normalizations,
\begin{equation}
\Rt(\st A,\st B) = \Rs(\rho,\sigma) \, .
\end{equation}
\end{proposition}
The proof can be found in Appendix \ref{app:comparable}.

\subsection{Example: testing linear polarizations}

In this paragraph we analyze the  incompatibility  of setups consisting in the preparation of a pure qubit state and in the measurement of a sharp observable on the output.
 For concreteness, we  refer to the   case of polarization qubits and   we analyze setups that probe the action of
an unknown  process on photons with a given linear polarization.

We consider two setups.  In one setup,  the input photon has linear polarization at $\theta/4$ degrees relative to the vertical axis and the output photon is measured with a polarizing filter at $\varphi/4$ degrees.   In the  other setup, the input photon  has linear polarization at $-\theta/4$ and the output photon is measured with a filter at $-\varphi/4$ degrees.  The factor 4 is just for later convenience.    The two setups are described by the   testers 
$\st A=\{A_1,A_2\}$, $\st B=\{B_1,B_2\}$,  defined by
\begin{align}
\label{eq:defABq}
A_1&=P_{-\varphi/2}\otimes P_{-\theta/2}, \quad &B_1&=P_{\varphi/2}\otimes P_{\theta/2},  \nonumber \\
A_2&=P_{\pi-\varphi/2}\otimes P_{-\theta/2}, \quad &B_2&=P_{\varphi/2-\pi}\otimes P_{\theta/2} \,
\end{align}
where we used the notation
\[  P_\alpha:=\frac{1}{2}(I+\sin\alpha\,\sigma_x+\cos\alpha\,\sigma_z) \,.  \]

\begin{figure}
\begin{center}
\includegraphics{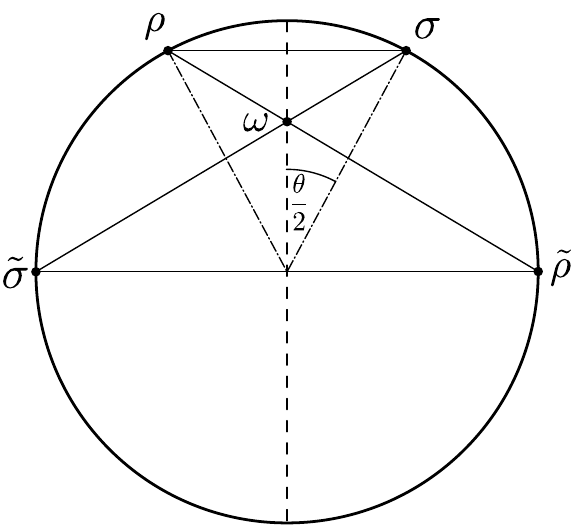}
\end{center}
\caption{\label{fig:qubit2} Illustration of the optimal choice of the normalizations $\widetilde\rho$, $\widetilde\sigma$ of the admixed testers $\st N^{(\st A)}$, $\st N^{(\st B)}$ for the incompatibility of testers defined in Eq.~(\ref{eq:defABq})}
\end{figure}

\begin{figure}
\begin{center}
\includegraphics[width=0.85\linewidth]{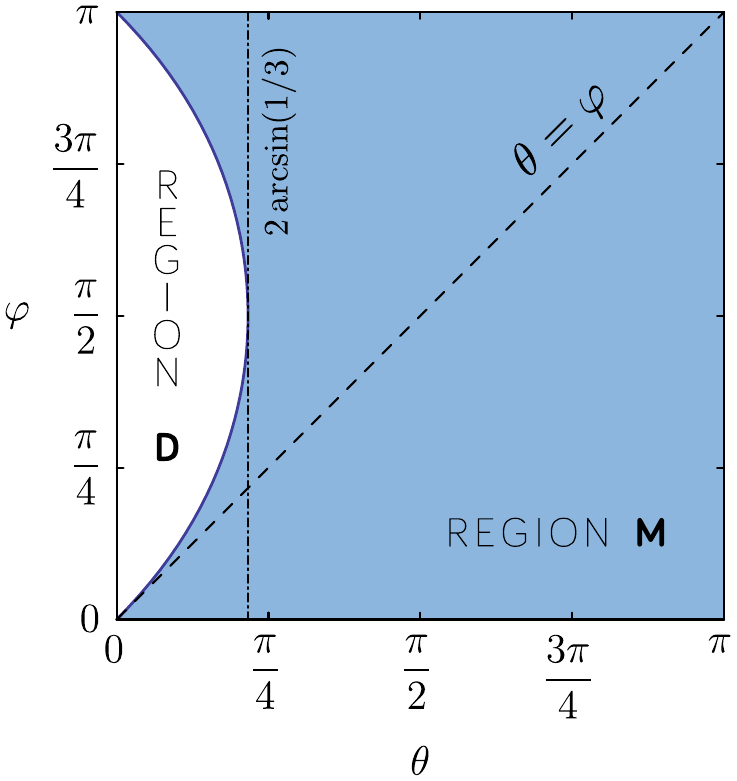}
\end{center}
\caption{\label{fig:regionM}
This figure illustrates the splitting of the parameter space
of two-outcome factorized rank-one testers into two regions $\mathsf M$ and $\mathsf D$
investigated separately. The robustness for testers is given
by robustness of normalization states in the (shaded blue) region $\mathsf M$
(Proposition \ref{prop:achievinM}). In the remaining
region ${\mathsf D}$ this is no longer the case as it is illustrated
in Fig.~\ref{fig:bothregions}.}
\end{figure}

Physically, $\theta/2$ and $\varphi/2$ represent the angles between the polarization filters preparing the input photon and measuring the output photon, respectively.
By varying the angles $\theta$ and $\varphi$ we obtain different pairs of testers,   which can be represented in a square, as  in Fig.~\ref{fig:regionM}.
Note that, unless $\theta  =  0$,  the two testers $\st A$ and $\st B$ are incompatible, because they have different normalization states.
Moreover, even testers with $\theta=  0$ may be incompatible, due to the incompatibility of the output measurements: for $\theta  =  0$, we have the equality
\[  \Rt(\st A, \st B)=\Rm(\st P, \st Q)  \, ,  \]
showing that the incompatibility of the testers is quantified by the incompatibility of the canonical POVMs associated with them.     Note that the  canonical POVMs are compatible only if $\varphi$ is equal to 0 or to $\pi$.

To start our analysis, we evaluate the lower bound arising from the incompatibility of the input states.  Combining Proposition~\ref{prop:normalization} with
Equation (\ref{eq:lambda_min}) we obtain
\begin{align}
\label{eq:minlforAB}
\Rt(\st A, \st B)
\geq \Rs(\rho, \sigma )
=\frac{\sin \frac{\theta}{2}}{1+\sin \frac{\theta}{2}}\, ,
\end{align}
where $\rho$ and $\sigma$ are the normalization states of $\st A$ and $\st B$, respectively.
Then, it is interesting to characterize the cases where the bound is saturated.
 To this purpose, we define a suitable parameter region, which we call ``region  $\mathsf M$":
 \begin{equation}
\mathsf M\equiv \Big\{(\theta,\varphi)\;|\;\theta,\varphi\in[0,\pi]\, ,\, \sin \frac{\theta}{2}\geq \frac{\sin \varphi }{2+\sin \varphi}\Big\}  \, .
\end{equation}
Note that the region $\mathsf M$ contains  all the testers
testing how well a given polarization is preserved, that is, all testers with $\varphi  =  \theta$.
 More generally, the points in region $\mathsf M$ are illustrated in  Fig.~\ref{fig:regionM}.

With a slight abuse of terminology, we say that two testers $\st A$ and $\st B$  \emph{belong to region $\mathsf M$} if the corresponding parameters belong to region $\mathsf M$.    We then have the following proposition.

\begin{proposition}
\label{prop:achievinM}
If a pair of testers $\st A,\st B$  belongs to region $\mathsf M$, then their incompatibility is    quantified by the incompatibility of the input states. In formula:
\begin{equation}
\label{eq:achieveinM}
\Rt(\st A,\st B)=  \Rs  (\rho, \sigma)  = \frac{\sin \frac{\theta}{2}}{1+\sin \frac{\theta}{2}} \, .
\end{equation}
\end{proposition}

The proof of this proposition is in Appendix \ref{app:achievinM}. A plot of the robustness of incompatibility is provided in Fig.~\ref{fig:bothregions}.

Outside region $\mathsf M$  the situation  is much trickier.  Here we do not have close formulas for the robustness of incompatibility and we had to resorted to  numerical  evaluation via semidefinite programming (SDP),
as outlined in Appendix \ref{app:SDP}.  The result of the evaluation is plotted in Fig.~\ref{fig:bothregions}. Conceptually,  our main findings are
\begin{enumerate}
\item  outside region $\mathsf M$,  the robustness of tester incompatibility is strictly larger than the robustness of state incompatibility
\item  for every fixed value of  $\theta$,  the maximum of the robustness is attained  when the angle $\varphi$ is equal to $\pi/2$, corresponding to mutually unbiased measurements on the output, or equivalently, to polarizing filters with a relative angle of 45 degrees.
\end{enumerate}

\begin{figure}
\begin{center}
\includegraphics[scale=0.8]{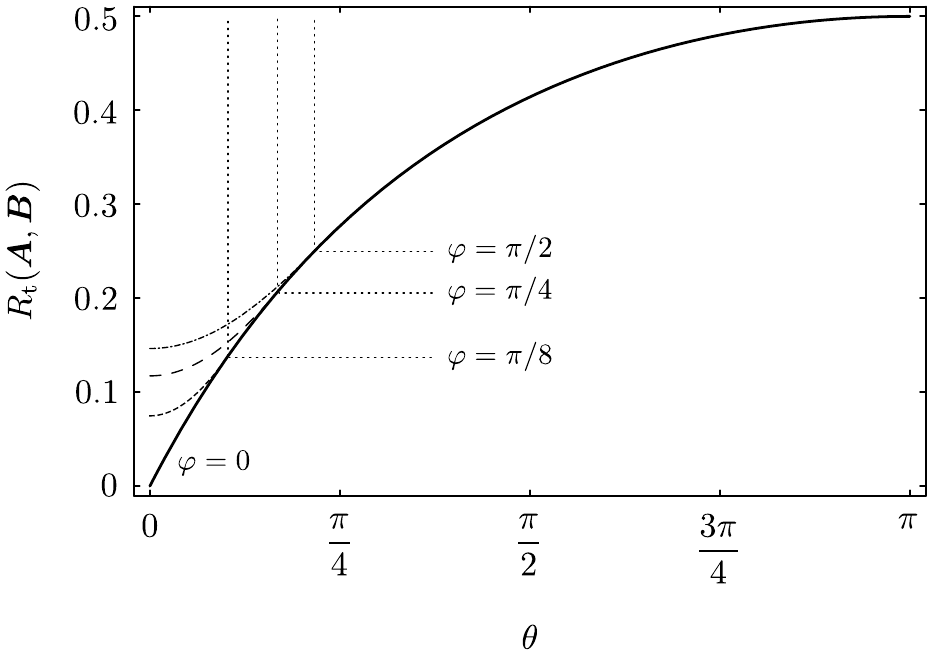}\\
\includegraphics[scale=0.8]{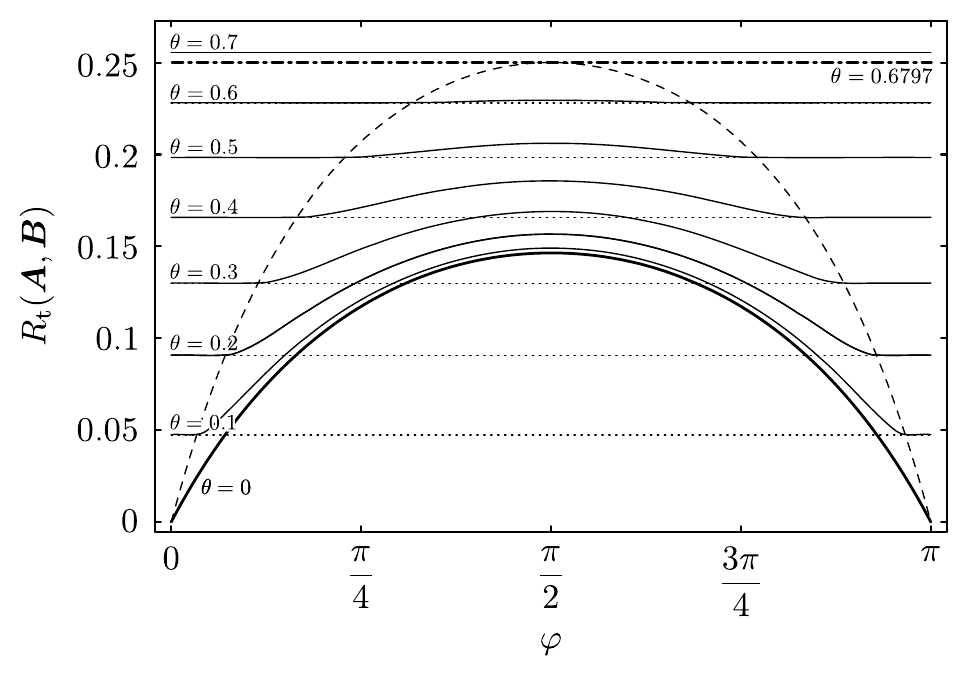}
\end{center}
\caption{\label{fig:bothregions} Robustness of incompatibility of two qubit testers as defined in Eq.~(\ref{eq:defABq}). Upper figure shows dependence on $\theta$ and situations for various choices of $\varphi$. Solid line depicts bound of Eq.~(\ref{eq:minlforAB}). Bottom figure shows dependence on $\varphi$ for various choices of $\theta$. Dotted lines represent Eq.~(\ref{eq:minlforAB}) which coincides with $\Rt(\st A,\st B)$ for $\theta \geq 0.6797$.}
\end{figure}

Before concluding,  we note that our results can be extended from testers involving sharp measurements to  testers involving arbitrary two-outcome POVM, provided that a suitable technical condition is satisfied.

\begin{proposition}
\label{prop:qABpovm}
Consider a pair of two-outcome testers $\st A$ and $\st B$ with operators
\begin{align}
\label{eq:defABqPOVM}
A_1&=E_1\otimes P_{-\theta/2}, &B_1&=F_1\otimes P_{\theta/2} , \nonumber \\
A_2&=E_2\otimes P_{-\theta/2}, &B_2&=F_2\otimes P_{\theta/2},
\end{align}
where $\{E_1,E_2\}$, $\{F_1,F_2\}$ are arbitrary qubit POVMs. Then,  one has the equality
\begin{equation}
\Rt(\st A,\st B)= \Rs(\rho,  \sigma )=\frac{\sin \frac{\theta}{2}}{1+\sin \frac{\theta}{2}} \, ,
\end{equation}
whenever the angle $\theta$ satisfies the condition   $\theta \geq 2 \arcsin (1/3)$.
\end{proposition}
The proof is presented in Appendix \ref{app:proofqABpovm}.

\section{Extension to measurement setups with  multiple time steps}
\label{sec:sec6}
So far we have discussed the compatibility of pairs of experiments designed to test quantum channels, however, the definitions and most of the results hold in a more general setting:  we can consider
\begin{enumerate}
\item  the compatibility of more than two  quantum testers
\item  the case of testers that probe quantum processes with multiple time steps.
\end{enumerate}
A quantum process with multiple time steps can be viewed as the quantum version of a causal network, consisting of an ordered sequence of  processes correlated by the presence of a quantum memory (see Figure \ref{fig:Ncomb} down).
\begin{figure*}
\begin{center}
\includegraphics[scale=1.2]{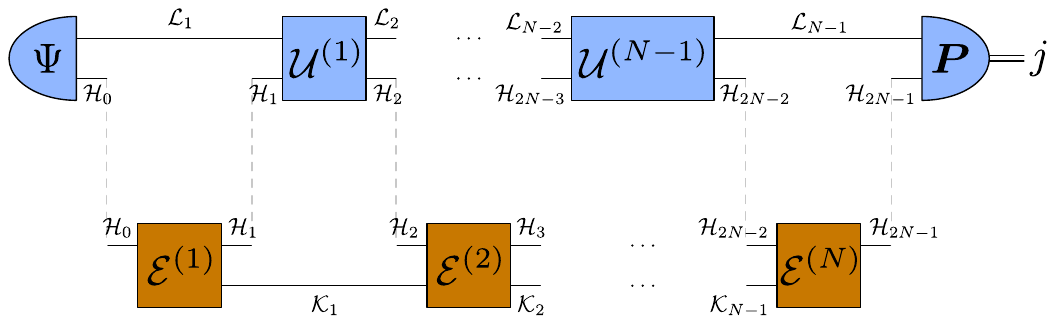}\\
\end{center}
\caption{Illustration of quantum causal network, or $N$-comb (lower orange) that can be ``measured'' by a quantum $N$-tester (upper blue).
Here $\spc H_i,  \, i\in \{0,1,\dots,2N-1\}$ are the Hilbert spaces describing the inputs and outputs of the network (and outputs and inputs of the $N$-tester), $\spc L_j ,   j\in  \{0,1,\dots,N-1\}$ are the Hilbert spaces of the internal memories of the network, $\spc K_l \, , l\in\{0,1,\dots,N-1\}$ are the Hilbert spaces of the internal memories of the $N$-tester.
$\Psi$ is an input state to the $N$-tester (which can be chosen to be pure without loss of generality), $\mathcal U^{(l)}$,  $l\in\{1,2,\dots,N-1\}$ are its quantum channels (which can be chosen to be unitary without loss of generality), and $\st P:=\{  P_j\}_{j\in J}$ is a POVM, representing a measurement on the final output systems of the network of tested processes $\mathcal E^{(j)},j\in\{1,2,\dots,N\}$.}
\label{fig:Ncomb}
\end{figure*}
We call  a  quantum causal network of such form a \emph{quantum comb} \cite{comb,dariano}  (also known as \emph{quantum strategy} in the context of quantum games \cite{gutwat}).
More specifically, we call \emph{quantum $N$-comb} a causal network consisting of $N$ steps.

Two $N$-combs are \emph{of the same type} if they have the same sequence of input/output Hilbert spaces.  Like quantum channels, quantum $N$-combs can be represented by Choi  operators. An operator  $\RC^{(N)}$ is a quantum N-comb (illustrated in Figure \ref{fig:Ncomb} down) if and only if it is positive and satisfies the equations
   \begin{align}
   \nonumber   \Tr_{2N-1}\left[ \RC^{(N)}\right]  & =     I_{2N-2}  \otimes  \RC^{(N-1)},\\
  \nonumber  \Tr_{2N-3} \left[ \RC^{(N-1)}  \right]   &  =  I_{2N-4}  \otimes  \RC^{(N-2)}, \\
   \nonumber   &  ~\, \vdots  \\
   \Tr_1 \left[ \RC^{(1)}  \right] & =   I_{0},
    \label{normalizationNcomb}
   \end{align}
where $\Tr_n$ denotes the partial trace  over the Hilbert space $\spc H_n$,
$I_n$ denotes the identity operator on $\spc H_n$ and $ \RC^{(n)} $
is a positive operator on  $\bigotimes_{i=0}^{2n-1} \spc H_{i}$.

In order to extract information about the quantum network we need to
perform an experiment  that provides inputs to the network and processes the outputs received at all time steps.    The most general way to interact with a quantum causal network is by connecting it with another causal network, as  illustrated in Figure \ref{fig:Ncomb}.

Again,   we will call the testing network a \emph{tester} \cite{opnorm,dariano} --- more specifically, we will call \emph{$N$-tester} a network designed to test processes of $N$ time steps.   Quantum testers have been recently considered in the tomographic characterization of non-Markovian evolutions  \cite{RiWoMo2015},  a scenario where probing  multitime quantum processes becomes highly relevant.

Two $N$-testers are \emph{of the same type} if they have the same sequence of input/output Hilbert spaces. An  $N$-tester illustrated in Figure \ref{fig:Ncomb} (up)  is described by a set of positive operators $\st T  =  \{   T_j\}_{j\in \set S}$,   where   $\set S$ is  the set of possible outcomes and each operator  $T_j$ acts  on the tensor product Hilbert space $\spc H_{\rm network}  =  \bigotimes_{i=0}^{2N-1}  \spc  H_i$.  The normalization condition for a tester is given by the  following set of equations
   \begin{align}
   \nonumber  \sum_{j\in \set S }  T_j  & =     I_{2N-1}  \otimes  \Theta^{(N)},\\
  \nonumber  \Tr_{2N-2} \left[  \Theta^{(N)}  \right]   &  =  I_{2N-3}  \otimes   \Theta^{(N-1)}, \\
   \nonumber   &  ~\, \vdots  \\
\nonumber      \Tr_2 \left[ \Theta^{(2)}  \right] & =   I_{1}\otimes \Theta^{(1)} ,\\
    \label{normalizationNtester}
    \Tr \left[\Theta^{(1)}\right]  &  =1 \, ,
   \end{align}
where $\Theta^{(n)}$ is a positive operator on
$\bigotimes_{i=0}^{ 2n-2}  \spc H_i$. We call such positive operator
$\Theta^{(N)}\equiv\Theta$ the \emph{normalization} of the tester $\st T$.

When quantum $N$-comb and quantum $N$-tester are combined the probability
of the outcome $j\in \set S$ is given by the \emph{generalized Born rule}
\begin{equation}
p_j=\tr{\left[ T_j   \,  \RC^{(N)}  \right]}\,.
\end{equation}
Different $N$-testers represent different (and possibly complementary) ways to extract information about a quantum $N$-comb. In the following we will formulate the elementary properties of compatibility for two or more $N$-testers. Let us stress that previous sections treat
the case of 1-testers, which we for simplicity denoted until now as testers and 1-combs traditionally called channels.

\begin{Def}
 Let $  \{\st T^{(x)} \,, \,   x\in \set X\}$ be a set of testers of the same type, the $x$-th tester with outcomes in the set $\set S_x$ and the normalization $\Theta_x$. The testers are compatible if there exists a joint tester ${\st C}  =  \{  C_{\bf{k}}\}_{\bf{k}\in \set S_1\times\dots\times \set S_{|X|}}$ such that for all $x\in X$ and for every $j_x\in \set S_x$
 \begin{align}\label{compatibilityKtesters}
  T^{(x)}_{j_x}   =   \sum_{{\bf k}:  k_x=  j_x} C_{\mathbf{k}}\,,
 \end{align}
where ${k}_x$ is  the $x$-th component of vector $\bf{k}$.
 \end{Def}

As in the case of compatibility of pairs of 1-testers, one can show that a set of testers  $\{\st T^{(x)}\}_{x\in X}$  are compatible only if they have the same normalizations. Moreover, Theorem \ref{prop:norm} generalizes to the following

\begin{proposition}
For each $N$-tester $\st T^{(x)}$ with normalization $\Theta_x$ define the \emph{canonical POVM} $\st P^{(x)}=\{ P^{(x)}_{j_x}\}$
\begin{equation}
P^{(x)}_{j_x}  =  \left(I_{2N-1} \otimes  \Theta_x^{-\frac 12}\right)  \,  T^{(x)}_{j_x}   \,  \left(I_{2N-1} \otimes  \Theta_x^{-\frac 12}\right) .
\end{equation}
The testers $\{  \st T^{(x)} \, ,  x\in \set X\}$ are compatible if and only if
\begin{enumerate}
\item their normalizations coincide ($\Theta_x\equiv\Theta$ for all $x$)
\item the canonical POVMs $\{  \st P^{(x)} \, , x\in \set X\}$ are compatible.
\end{enumerate}
\end{proposition}

\begin{proof}
The proof is a direct generalization of the proof of Theorem
\ref{prop:norm}.
\end{proof}

The incompatibility of multitime testers can be quantified in the same way as we did in the  $N=1$  case.    Again, the idea is to measure the incompatibility of a set of testers based on the amount of ``noise'' that one has to add in order to make them compatible.

\begin{Def}
The testers $\{\st T^{(x)},x\in \set X\}$ are $\lambda$-compatible if
for any $x$ there exists a tester $\widetilde{\st T}^{(x)}$ with the
outcome set $\set S_x$ such that the testers $\{(1-\lambda )  \st T^{(x)}  +
\lambda \widetilde{\st T}^{(x)}, x\in \set X\}$ are compatible.
\end{Def}

\begin{Def}
The \emph{robustness of incompatibility} of a set of testers $\mathsf T:=\{\st T^{(x)}\}_{x\in \set X}$, denoted by $\Rt(\mathsf T)$,  is  the minimal $\lambda$ such that the testers in the set $\mathsf T$ are $\lambda$-compatible.
\end{Def}

Note that every set of testers is $\lambda$-compatible with $\lambda =  1-|\set X|^{-1}$, as one can see from a simple adaptation of Proposition \ref{prop:maxlambda}: essentially, one can always make the testers compatible by uniformly mixing them.    In other words, we have the bounds
\[0\leq \Rt  (\set T)\leq 1-|\set X|^{-1}  \, , \]
valid for every set $\set T$ of testers  containing   $|\set  X| $ elements.
We now show  that the upper bound can be saturated.   To this purpose, we  formulate a lower bound in terms of normalization operators $\Theta_x$. The crucial observation is that the normalization operator $\Theta_x$ is an $N$-comb---that is, it satisfies the conditions in  Eq.~(\ref{normalizationNcomb}).    Physically, this means that $\Theta_x$ represents a quantum causal network, consisting of a sequence of $N$ time steps.    The distinguishability of the causal networks associated with the original testers gives a lower bound to  the robustness of incompatibility:

\begin{proposition}
Let $ \set T= \{\st T^{(x)} \,, \,   x\in X\}$ be a set of testers of the same type
with normalizations $\Theta_x$, respectively. Then the
robustness of incompatibility of this set is lower bounded by
\begin{equation}
\Rt (\set T)\ge    1-  \frac 1{  |X|  \, p_{\rm succ}},
\end{equation}
where $p_{\rm succ}$ is the maximum probability of success in distinguishing among the quantum causal networks associated with operators $ \left\{  \Theta_x  , \, x\in X\right\}$.  In particular, the
bound is saturated whenever  the networks are perfectly distinguishable.
\end{proposition}

\begin{proof}
Assume the testers are  $\lambda$-compatible for a certain $\lambda$ and let $\st C=\{C_{\bf{j}}\}$ be the joint tester that guarantees the compatibility.  Let us denote by $\Theta$ the normalization of the joint tester. Then necessarily for all $x\in X$
\begin{align}\label{dual}
\Theta   \ge  (1-\lambda) \, \Theta_x  \, .
\end{align}
In order to compute the robustness, we have to minimize $\lambda$ over all operators $\Theta$ subject to the constraint that $\Theta$ is the normalization of a joint tester satisfying the compatibility condition.    We now relax this constraint and assume only that $\Theta$ satisfies the normalization conditions in Eq.~(\ref{normalizationNtester}).
Defining $\mu   :=   (1-\lambda)^{-1}   |X|^{-1}$, we have that minimizing $\lambda$ under the condition (\ref{dual}) is equivalent to minimizing $\mu$ under the condition
\begin{align}
\mu  \,     \Theta  \ge     \frac 1{|X|}  \,   \Theta_x, \qquad \forall x\in X \, .
\end{align}
Now,  the minimization of $\mu$ under the condition that $\Theta$ is a comb is a semidefinite program.  This semidefinite program was recognized in  Ref.~\cite{combSDP}  as  the dual to the maximization of the success probability in the discrimination of the networks  $\{  \Theta_x \, ,  x \in\set X\}$.  More precisely,  Theorem 1 of Ref.~\cite{combSDP}  guarantees that  the minimum of $\mu$  is  equal to the maximum probability of success  $p_{\rm succ}$.   Hence, we must have     $\mu   \ge  p_{\rm succ}$, or equivalently
\[  \lambda  \ge 1  -    \frac 1 { |X|  p_{\rm succ}}  \, . \]       Since the inequality  holds for every $\lambda$, it must hold  also for the minimum $\lambda$,   leading to the bound  $\Rt   (\set T)\ge 1  -    1/( |X|  p_{\rm succ})$.   If the quantum networks  $\{ \Theta_x \, ,x\in X\}$   are perfectly distinguishable, one has $p_{\rm succ} =1$ and, therefore, $\Rt  (\set T)\geq 1- 1/|X|$. On the other hand, we already mentioned that every set of testers is $\lambda$-compatible with $\lambda =  1-1/|X|$, which concludes the proof.
\end{proof}

In the case of two testers  ($|\set X|=2$), the above bound has a nice expression in terms of the \emph{operational distance} between two quantum causal networks \cite{opnorm,dariano,combDiscrimination}. Specifically, one has
\begin{equation}
p_{\rm succ}  =  \frac 12    \left(   1  +    \frac 12   \left\|   \Theta_1   -     \Theta_2   \right\|_{\rm op}  \right) ,
\end{equation}
where $\|  \cdot \|_{\rm op}$ is the operational norm  \cite{opnorm,dariano,combDiscrimination}.  Inserting this expression in the lower bound we  then obtain
\begin{equation}
\label{eq:generalbound}
\Rt \ge       \frac{\left\|   \Theta_1   -     \Theta_2   \right\|_{\rm op}   }{2+  \left\|   \Theta_1   -     \Theta_2   \right\|_{\rm op}  } \, .
\end{equation}
Let us note that in the $N=1$  case  the operational norm coincides with the trace norm, which implies that Eqs.~\eqref{eq:lambda_mingen} and (\ref{eq:generalbound}) match.

\section{Conclusions}
\label{sec:sec7}

In this paper we have introduced the notion of compatibility for measurement setups designed to probe quantum dynamical processes.  We highlighted how the time structure, with its division into inputs and outputs, affects the notion of compatibility.    In particular, we highlighted the differences between the compatibility of ordinary quantum measurements, described by POVMs, and more general setups with multiple time steps, described by quantum testers.

When testing processes consisting of a single time step, the key differences between POVMs and testers are the following:
\begin{enumerate}
\item  For POVMs, commutativity implies compatibility. For testers, the implication is false, whenever the testers entail preparations of  distinct input states.  In this way, even testers consisting of mutually commuting projectors can turn out to be incompatible.
\item For POVMs, the  maximum amount of incompatibility that can be found in a finite dimensional system increases with  the dimension, reaching the largest value only for infinite dimensional systems   \cite{heinosaari,review}.  For testers, the maximum amount of incompatibility     is the same for all testers with non-trivial input: two testers are maximally incompatible whenever they entail the preparation of orthogonal input states.
\item For two-outcome POVMs, the ability to compare one element of a POVM with one element of the other POVM implies compatibility. For testers, the implication is false: two comparable testers may not be compatible. Also in this case, the incompatibility originates in the incompatibility of input states.
\end{enumerate}

Physically, the differences  arise from the fact that the incompatibility of testers can arise from a different source than the incompatibility of POVMs.   Such a  different source is the incompatibility of  the input states: essentially,   probing a process on a certain input precludes the possibility of probing the process on another input.
  As a result, the incompatibility of two testers can arise from two contributions: the incompatibility of the input states sent to the process and the incompatibility of the measurements performed on the output.    To quantify these contributions, we provide lower and upper bounds on the tester incompatibility in terms of the state and measurement incompatibilities, respectively.

For simplicity, we illustrated most of our results in the case of testers probing processes consisting of a single time step.
However,  all results can be generalized to testers that probe quantum processes consisting of multiple time steps.
  In particular, we showed that two general  testers can be incompatible, because they use two distinct sequences of interactions in order to probe an unknown multiple time step process. The distinguishability of the sequences of interactions provides a lower bound to the incompatibility of the resulting testers  and, again, maximum incompatibility is obtained when the sequences are perfectly distinguishable.

Since the incompatibility of ordinary measurements is a resource in several applications (steering, device independent cryptography, etc.), we believe that the research program on the incompatibility of testers, initiated in this paper, will have an  impact on the design of new quantum protocols.  At a more fundamental level, the study of the dynamical  properties of quantum  causal networks is expected to shed light on foundational questions about time and causal structure in quantum theory.

\section*{Acknowledgments}
We thank the anonymous referees for stimulating us to make the paper more self contained and more readable for a wider group of readers. We thank also Anna Jen\v cov\'a for stimulating discussions. We acknowledge support from the SRDA grant APVV-0808-12 (QETWORK),
VEGA Grant No.~2/0125/13 (QUICOST), from the the Foundational Questions Institute (FQXi-RFP3-1325), from the National Natural Science Foundation of China through Grants 11450110096 and 11350110207, from the 1000 Youth Fellowship Program of China, and from the HKU Seed Funding for Basic Research.
DR was supported via SASPRO Program No.~0055/01/01 (QWIN) and
MS acknowledges support by the European Social Fund and the state budget of the Czech Republic under Operational Program Education for Competitiveness (Project
No.~CZ.1.07/2.3.00/30.0004) and by the Development Project
of Faculty of Science, Palacky University. MZ acknowledges
the support of Czech Science Foundation (GA {\v C}R) project
No.~GA16-22211S.

\appendix

\section{Proof of proposition \ref{prop:rhoT}}\label{app:proofrhoT}

Given a physical implementation $\map T  =  (\spc H_{\rm anc},  \Psi,  {\st P})$,  the corresponding tester $\st T$ is given by Eq.~(\ref{Tj}).     Using this expression, we obtain
\begin{align}
\nonumber \sum_j   T_j  &  =  \sum_j     \Tr_{\rm anc} [     (  P_j  \otimes  I_0 )  \,   (   I_1 \otimes   {\tt SWAP}  \Psi^{T_0}  {\tt SWAP} )] \\
  &  =    \Tr_{\rm anc} [           I_1 \otimes    \Psi^{T_0}  ] =   I_1  \otimes \rho  ,
\end{align}
where $\rho:  =      \Tr_{\rm anc} [            \Psi^{T_0}  ]$.
\qed

\section{Proof of proposition \ref{prop:ancillafree}}\label{app:proofrhoT2}
Let $ \map T   =  (\mathbb C,  \Psi,  {\st P})$ be an ancilla-free implementation.  Then, Eq.~(\ref{Tj}) gives
\begin{equation}
 T_j  =   (  P_j  \otimes  I_0 )  \,   (   I_1 \otimes     \Psi^{T} )=   P_j\otimes \Psi^T  \, ,\quad \forall j .
\end{equation}
The normalization of the tester reads
\begin{equation}  \sum_{j}   T_j     =  \left( \sum_j  P_j   \right) \otimes  \Psi^T =  I_1  \otimes \Psi^T  \, ,
\end{equation}
which, compared with the normalization condition (\ref{normalization}), implies $\rho  =   \Psi^T$.

Conversely, suppose that the tester operators are of the form $  T_j  =  P_j\otimes \rho$.
Setting $\Psi  :=  \rho^T$, it is immediate to check that $ \map T  =   (\mathbb C,  \Psi,  {\st P})$ is an ancilla-free implementation.   \qed

\section{Further examples of testers}\label{app:examples}

\subsubsection{Testers with  classical ancilla}

In the case of ancilla-free testers, there are no correlations between the  state sent as input to the unknown process  and the measurement performed on its output.     For some applications,  like quantum process tomography  \cite{Nielsen00}, it is important to test the action of the process on multiple input states,  keeping track of which state has been used as a probe.
This task can be accomplished by using a classical system as ancilla.

Schematically, we can consider a setup as in Figure \ref{fig:producttesterscheme}, where the input state is prepared in a state $\rho_k$, correlated with a classical random variable, which assumes the value $k$ with probability $q_k$.    Then,  the output is measured with a POVM ${\st P}^{(k)}   =   \left\{   P^{(k)}_{j_k}  \,    ,  j_k  \in  \set S_k \right\}$, whose outcome set   $\set S_k$ possibly depends on the index $k$.      This dependence allows to keep track of the  value of $k$, as different values of $k$ can correspond to disjoint sets of outcomes.

\begin{figure}
\begin{center}
\includegraphics{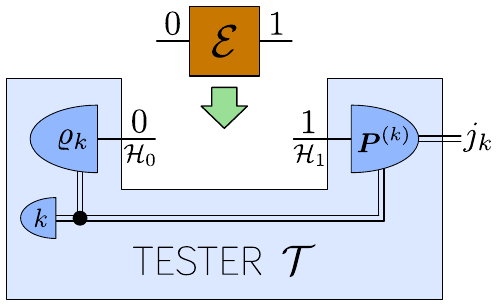}
\end{center}
\caption{\label{fig:producttesterscheme}
Diagrammatic representation of a tester with classical ancilla.
}
\end{figure}

Note that the joint state of the input system and the classical random variable $k$ can be represented  as    the quantum-classical state
 \begin{align}\label{qcstate}
 \Psi  =  \sum_k  \,  q_k  \,    \rho_k \otimes  |k\>\< k|  \, ,
 \end{align}
where  $  \{  |k\>   \}$ is an orthonormal set of states of a suitable ancilla.   In this picture, a measurement that depends on $k$ is just a joint measurement on the system and the  ancilla.
 This observation  motivates the following
  \begin{Def}
  We say that a tester $\st T$ can be implemented with a  classical ancilla if there exists an implementation $\map T  =  (\spc H_{\rm anc} \, ,  \Psi,  {\st P})$ such that the state $\Psi$ is quantum-classical.
  \end{Def}
Testers that can be implemented with a classical ancilla have the form ${\st T}   =  \{  T_{ j_k}  \}$, with
\begin{align}\label{classicalanc}
T_{j_k}    =  \sum_k q_k\,    P^{(k)}_{j_k}  \otimes      \rho_k     \, .
\end{align}
Mathematically, the set of  testers with classical ancilla is nothing but the convex hull of the set of ancilla-free testers.
This fact can be easily seen  by comparing Eq.~(\ref{classicalanc})  with Eq.~(\ref{ancillafree}).

\subsubsection{Testers with genuinely quantum ancilla}

Some testers cannot be implemented with a classical ancilla.   This is the case, for example, of testers containing non-separable   operators --- i.e.~operators that cannot be written in the form
\[  T_j  =   \sum_k    \,  A_{jk}  \otimes  B_{jk} \]
with positive  $A_{jk}$ and $B_{jk}  $.
A concrete example arises when testing how well a process preserves the maximally entangled state  $|\Phi_+\>$.   To this purpose, we can prepare the input and the ancilla in the state $\Psi  =  |\Phi_+\>\<  \Phi_+| $, apply the process $\map E$ on the first system,  and then measure the output and the ancilla with two-outcome POVM $  {\st P}  =\{  P_1,P_2\}$ with
\begin{equation}
P_1  =     |\Phi_+\>\<\Phi_+|    \, ,   \quad P_2  =    I\otimes I   -|\Phi_+\>\<\Phi_+|  .
\end{equation}
This setup corresponds to the two-outcome tester  ${\st T}   =   \{T_1,  T_2\}$ with
\begin{equation}
T_1  =   \frac  1 d \,      |\Phi_+\>\<\Phi_+|      \, ,  \quad  T_2  =    \frac  1 d   \,   \left(   I\otimes I    - |\Phi_+\>\<\Phi_+|  \right)  .
\end{equation}
This tester cannot be realized with a classical ancilla, because the operator $T_1$ does not have the separable form of Eq.~(\ref{classicalanc}).

 \section{Proof of theorem \ref{prop:norm}}\label{app:proofnorm}

Suppose that ${\st A}  =  \{ A_j\}  $ and  ${\st B}=   \{  B_k\}$ are compatible. Then,   there exists a tester ${\st C}=  \{  C_{jk}\}$   such that the compatibility condition \eqref{compatibilitytester} holds.   Note that Eq.~(\ref{compatibilitytester}) implies that $\st A$, $\st B$, and $\st C$ have the same normalization state, call it $\rho$.   Now, define the canonical POVMs associated with $\st A$, $\st B$, and $\st C$, namely the POVMs ${\st P}  =  \{  P_j\}$,  ${\st Q}  =  \{   Q_k\}$, and ${\st R} =  \{  R_{jk}\}$ with operators
\begin{align}
P_{j} &= \left(I\otimes \rho^{-\frac12}\right)\,A_{j}\, \left(I\otimes \rho^{-\frac12}\right),  \nonumber\\
Q_{k}&= \left(I\otimes \rho^{-\frac12}\right)\,B_{k}\, \left(I\otimes \rho^{-\frac12}\right),  \\
R_{jk}&= \left(I\otimes \rho^{-\frac12}\right)\,C_{jk}\, \left(I\otimes \rho^{-\frac12}\right) .\nonumber
\end{align}
Then, it is immediate to obtain the compatibility conditions for the canonical POVMs $\st P$ and $\st Q$.  Indeed, one has
\begin{align}
\sum_{k} R_{jk}&=   \left(I\otimes \rho^{-\frac12}\right) \left(\sum_k C_{jk}\right) \left(I\otimes \rho^{-\frac12}\right)\nonumber\\
&  =  \left(I\otimes \rho^{-\frac12}\right)  A_j \left(I\otimes \rho^{-\frac12}\right) \\
& \equiv  P_j \, ,\nonumber
\end{align}
and similarly,
\begin{align}
\sum_{j} R_{jk}&=   \left(I\otimes \rho^{-\frac12}\right) \left(\sum_j C_{jk}\right) \left(I\otimes \rho^{-\frac12}\right)\nonumber\\
&  =  \left(I\otimes \rho^{-\frac12}\right)  B_k \left(I\otimes \rho^{-\frac12}\right) \\
& \equiv  Q_k \, . \nonumber
\end{align}

Conversely, if the normalization states are the same and the canonical POVMs $\st P$ and $\st Q$ are compatible, then one can use the joint POVM ${\st R}$  to define the  tester ${\st C}$ with operators
\begin{equation}
C_{jk}   :=  \left(I\otimes \rho^{\frac12}\right) R_{jk} \left(I\otimes \rho^{\frac12}\right) .
\end{equation}
 By construction, $\st C$ is a joint tester for $\st A$ and $\st B$: indeed, one has
\begin{align}
\sum_{k} C_{jk}&= \left(I\otimes \rho^{\frac12}\right)   \left(  \sum_k R_{jk}   \right)\left(I\otimes \rho^{\frac12}\right)  \nonumber\\
& =  \left(I\otimes \rho^{\frac12}\right) P_j \left(I\otimes \rho^{\frac12}\right)\\
&= A_j  \, ,\nonumber
\end{align}
having used the fact that $\rho^{\frac 12}$ is invertible on its support.  Similarly, one has
\begin{align}
\sum_{j} C_{jk}&= \left(I\otimes \rho^{\frac12}\right)   \left(  \sum_j R_{jk}   \right)\left(I\otimes \rho^{\frac12}\right)  \nonumber\\
& =  \left(I\otimes \rho^{\frac12}\right) Q_k \left(I\otimes \rho^{\frac12}\right)\\
&= B_k  \, .\nonumber
\end{align}
This concludes the proof.\qed

\section{Proof of proposition \ref{prop:common_noise}}
\label{app:nostrong}

 Let $\rho$, $\sigma$, and $\tau$ be the normalization states of $ \st A$,  $\st B$, and $\st N$, respectively.     Then, the compatibility  condition    implies  the relation
\[   (1-\lambda)  \,  \rho  +  \lambda    \,  \tau   =  (1-\lambda) \,   \sigma  +  \lambda\,  \tau  \, , \]
(cf. Proposition \ref{prop:necessary_normalization}).  This relation can be satisfied only if $\lambda = 1$ .

\section{Proof of proposition \ref{prop:diagonal}}\label{app:diagonal}
Since $\st A$ and $\st B$ are diagonal in the same basis, also the operators
$I_1\otimes \rho$ and $I_1\otimes \sigma$ are diagonal in the same basis.  As a result, also the canonical POVMs  $\st P$ and $\st Q$, defined by
\begin{align}
P_j&=   \left(I_1  \otimes \rho^{-\frac 12}  \right) \,  A_j \,  \left(I_1  \otimes \rho^{-\frac 12}  \right),\nonumber  \\
Q_k   &=   \left(I_1  \otimes \sigma^{-\frac 12}  \right) \,  B_k \,  \left(I_1  \otimes \sigma^{-\frac 12}  \right)  \, .
\end{align}
are  diagonal in the same basis.
Now, let $(\lambda \, , \widetilde \rho \, ,\widetilde \sigma)$ be a triple satisfying Eq.~(\ref{eq:ncnorm}) and let  $\omega$ be the state
\begin{align}
\label{eq:appnormt}
\omega & =  (1-\lambda) \, \rho+\lambda \,  \widetilde{\rho}\nonumber\\
&\equiv(1-\lambda)  \, \sigma+\lambda \, \widetilde{\sigma}   \, .
\end{align}
Without loss of generality, we can choose the operators $I_1  \otimes  \widetilde \rho $,   $I_1\otimes \widetilde \sigma$, and $I_1 \otimes  \omega$ to be diagonal in the same basis,
since $\widetilde \rho $, $\widetilde \sigma$ can be made diagonal by taking only their diagonal elements, which preserves validity of Eq.~(\ref{eq:appnormt}) and does not change $\lambda$.
Using this fact, we  define the testers  ${\st N}^{(\st A)}$ and ${  \st N}^{(\st B)}$ as
\begin{align}
N^{(\st A)}_j &: =   \left(I_1  \otimes \widetilde \rho^{\frac 12}  \right) \, \tilde P_j \, \,  \left(I_1  \otimes \widetilde \rho^{\frac 12}  \right),\nonumber \\
N^{(\st B)}_k &: =   \left(I_1  \otimes \widetilde \sigma^{\frac 12}  \right) \, \tilde Q_k \, \,  \left(I_1  \otimes \widetilde \sigma^{\frac 12}  \right) ,
\end{align}
where $\{\tilde P_j\}$, $\{\tilde Q_k\}$ are arbitrary POVMs diagonal in the common basis.
Note that also  ${\st N}^{(\st A)}$ and ${  \st N}^{(\st B)}$ are diagonal in the same basis as $\st A$ and $\st B$.   Now, by construction the testers   $  (1-\lambda) {\st A}  +  \lambda   {\st N}^{(\st A)}$ and $  (1-\lambda) {\st B}  +  \lambda   {\st N}^{(\st B)}$ have the same normalization state, namely $\omega$.    Moreover, their canonical POVMs $\overline P$ and $\overline Q$, defined by
 \begin{align}
\overline P_j &: =   \left(I_1  \otimes \omega^{-\frac 12}  \right) \,  \left[   (1-\lambda) {A}_j  +  \lambda   {N}^{(\st A)}_j   \right]  \,  \left(I_1  \otimes  \omega^{-\frac 12}  \right),\nonumber  \\
\overline Q_k &: =   \left(I_1  \otimes \omega^{-\frac 12}  \right) \,  \left[   (1-\lambda) {B}_k  +  \lambda   {N}^{(\st B)}_k   \right]  \,  \left(I_1  \otimes  \omega^{-\frac 12}  \right),
\end{align}
are also diagonal in the same basis.  Hence, they can be jointly measured.  Using Theorem \ref{prop:norm}, we conclude that the testers $\st A$ and $\st B$ are $\lambda$-compatible, whenever their normalization states $\rho$ and $\sigma$ are $\lambda$-compatible.  Taking the minimum over $\lambda$ we finally obtain the desired result.
\qed

\section{Proof of proposition \ref{prop:discrimination}}\label{app:discrimination}

Eq.~(\ref{eq:ncnorm}) can be  rewritten as
\begin{equation}
\widetilde\rho-\widetilde\sigma  =  \left(\frac1\lambda  -1\right)\,  ( \sigma - \rho),
\end{equation}
which implies that the operators $\widetilde \rho -  \widetilde \sigma$ and $\sigma-\rho$ must be proportional to one another, with the proportionality constant
\begin{equation}
\frac1\lambda  - 1      =   \frac{\| \widetilde \sigma  -\widetilde \rho \|}{\| \sigma - \rho\|}.
\end{equation}
Clearly, the minimum value of $\lambda$ is attained when the norm $\| \widetilde \rho  -\widetilde \sigma \|$ is maximal, compatibly with the requirement that  $\widetilde \rho -  \widetilde \sigma$ and $\rho-\sigma$ be proportional.
We now show that one can always choose  $\widetilde \rho$ and $\widetilde \sigma$ so that the norm has the maximum possible value, namely  $\| \widetilde \rho  -\widetilde \sigma\| = 2$.

To this purpose, we define the operator
\begin{equation}
\Delta:=\frac{\sigma - \rho}{\|\sigma - \rho\|} .
\end{equation}
Since $\Delta$  is self-adjoint, it can be decomposed as
\begin{equation}
\Delta=\Delta_+ - \Delta_- ,
\end{equation}
with $\Delta_+\ge 0$, $\Delta_-  \ge 0$, and $\Delta_+ \Delta_-  =  0$.
Moreover, $\Delta$ satisfies $\Tr[\Delta]  =  0$ and $\| \Delta \|  = 1$, which imply
\begin{align}
\tr  [ \Delta_+ -  \Delta_-]  =  0  \qquad {\rm and}  \qquad
  \tr  [ \Delta_+ +  \Delta_-]  =  1 \, ,
  \end{align}
or equivalently,  $\tr [\Delta_+]=\tr[ \Delta_-]=1/2$.
Hence, we can  define the density operators
\begin{align}
\widetilde\rho:=2 \Delta_+   \qquad {\rm and}    \qquad \widetilde\sigma:=2 \Delta_-  ,
\end{align}
which satisfy
\begin{equation}
\widetilde \rho  - \widetilde \sigma   =     2  \Delta     =   \frac { 2  (\sigma - \rho) }{\| \sigma - \rho\|}.
\end{equation}
In other words,   $  \widetilde \rho  - \widetilde \sigma$ and $\sigma - \rho$ are proportional and the proportionality constant is
\begin{equation}
\frac2{\|  \rho-\sigma\| }    =   \frac 1 {\Rs  (\rho,\sigma)} -1.
\end{equation}
In conclusion, we obtained
\begin{equation}\Rs(\rho,\sigma)  =    \frac{ \|  \rho-\sigma\|}{   \|\rho   -  \sigma\| + 2 }.
\end{equation}
\qed

The above proof has a nice geometric interpretation, highlighted in Figure   \ref{fig:qubit1}.
\begin{figure}
\begin{center}
\includegraphics{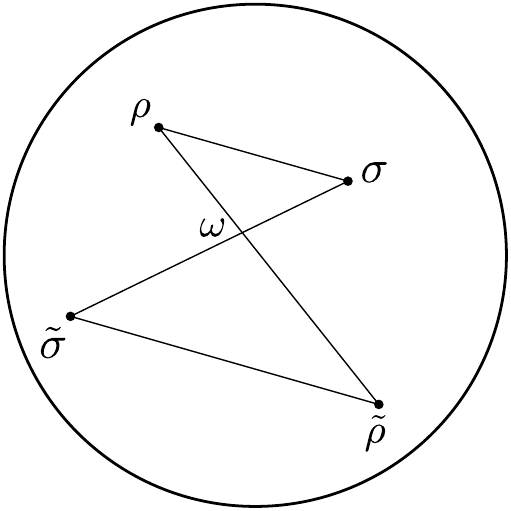}
\end{center}
\caption{\label{fig:qubit1} To obtain common normalization $I\otimes \omega$ for the same $\lambda$, lines connecting $\rho$ with $\sigma$ and
$\widetilde\rho$ with $\widetilde\sigma$ must be parallel.}
\end{figure}
Geometrically, the convex combination $\omega  = (1-\lambda)\rho+\lambda \widetilde{\rho}$ represents a point in the segment joining $\rho$ and $\widetilde \rho$.   Measuring the length of segments with the trace distance
\begin{equation}
d(\rho, \sigma)  =  \|  \rho  - \sigma\|,
\end{equation}
one has
\begin{align}
\nonumber  \lambda   &=   \frac{  d   (   \rho, \omega)}{d  (\rho  ,  \widetilde \rho)}   \\
 \nonumber  &  =  \frac{  d   (   \rho, \omega)}{d  (\rho  , \omega )   + d  (\omega, \widetilde \rho)    } \\
  \label{lambda}   &  =  \frac{  1  }{1  + \frac{d  (\omega, \widetilde \rho)}{d(\omega,\rho) }   }      \, .
\end{align}
Now, the relation
\begin{equation}
(1-\lambda)\rho+\lambda \widetilde{\rho}=(1-\lambda)\sigma+\lambda \widetilde{\sigma}
\end{equation}
implies that
\begin{enumerate}
\item the points $\rho,\sigma,\widetilde{\rho}$, and $\widetilde{\sigma}$ belong to the same plane ,
\item  the point  $\omega$
is the intersection of the segment joining $\rho$ and $\widetilde \rho$ with the segment joining $\sigma$ and $\widetilde \sigma$, and
\item   the triangles  with vertices $(\rho, \sigma, \omega)$ and $(\widetilde \rho, \widetilde \sigma,  \omega)$ are similar.
\end{enumerate}
Using the similarity of the triangles,  Eq.~(\ref{lambda}) becomes
\begin{align}
\lambda   =    \frac{  1  }{1  + \frac{d  (\widetilde \rho, \widetilde \sigma)}{d(\rho,\sigma) }   } .
\end{align}
Now, the distance  $d  (\widetilde \rho, \widetilde \sigma)$ is maximized by choosing $\widetilde \rho$ and $\widetilde \sigma$ to be as far as possible, 
but compatibly with the condition that $\widetilde \rho$ and $\widetilde \sigma$ must be states.
Hence, we have the bound
\begin{align}
\Rs  (\rho,\sigma)   \ge     \frac{  1  }{1  + \frac{2}{d(\rho,\sigma) }   }  \ge     \frac{  \|  \rho  -\sigma \|  }{\|\rho-\sigma\|  +2    }     ,
\end{align}
with the equality if and only if there exist states $\widetilde \rho$ and $\widetilde \sigma$ at distance $d(\widetilde \rho,  \widetilde \sigma)=2$.

\begin{figure}
\begin{center}
\includegraphics{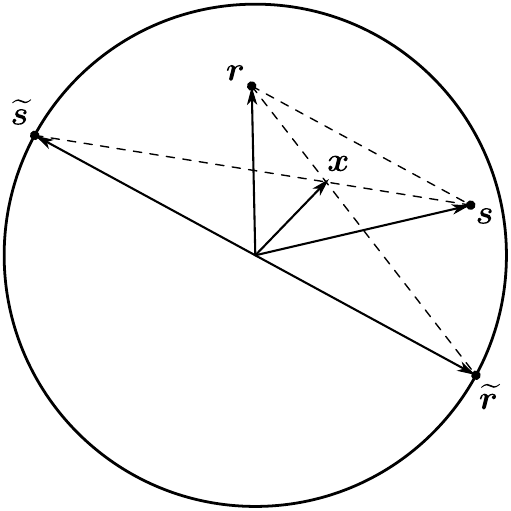}
\end{center}
\caption{\label{fig:qubitdiameter} In qubit case the common normalization given by Bloch vector $\st x$ points at the intersection of lines connecting Bloch vectors $\st r$ with $\st{\widetilde r}$ and $\st s$ with $\st{\widetilde s}$, while the vectors $\st{r}-\st{s}$ and $\st{\widetilde r}-\st{\widetilde s}$ must be parallel.}
\end{figure}

In the qubit case, these states can be easily found by exploiting the geometry of the Bloch sphere.
Indeed, the similarity of the triangles with vertices  $(\rho, \sigma, \omega)$ and $(\widetilde \rho, \widetilde \sigma,  \omega)$ implies that the segment joining $\widetilde \rho$ with $\widetilde \sigma$ should be parallel to the segment joining $\rho$ with $\sigma$.
Hence,  we can maximize the length $d  (\widetilde \rho,\widetilde \sigma)$ by choosing $\widetilde \rho$ and $\widetilde \sigma$ to be the extreme points of the diameter parallel  to the segment joining $\rho$ with $\sigma$, as in Figure \ref{fig:qubitdiameter}.
In terms of the Bloch vectors $\st r, \st s, \st {\widetilde r}, \st {\widetilde s} $ associated with the states $\rho, \sigma, \widetilde \rho, \widetilde \sigma$, the condition that the segments are parallel reads
\begin{equation}
\frac{ \st {\widetilde r}   -  \st  {\widetilde s} }{\|     \st {\widetilde r}   -  \st  {\widetilde s}  \|}  =  \frac{  \st s-\st r}{\|  \st r - \st s\|}.
\end{equation}
Choosing $ \st {\widetilde r}$ and $  \st  {\widetilde s}$ on the diameter, the above condition yields
\begin{equation}
\st {\widetilde r} =  \frac{  \st s-\st r}{\|  \st r - \st s\|}  \qquad {\rm and}  \qquad          \st {\widetilde s} =  \frac{  \st r-\st s}{\|  \st r - \st s\|}.
\end{equation}

For quantum systems of higher dimensions, the existence of orthogonal states $\widetilde \rho$ and $\widetilde \sigma$    does not follow directly from the geometric picture, but can be shown from the  spectral properties of the density matrices, as we did  in the proof at the beginning of this Appendix.

\section{Conjecture on the maximum amount of measurement-induced incompatibility}\label{app:fourier}

Currently, it is not known which pair of POVMs are the most incompatible, although a natural conjecture is that the maximum amount of incompatibility is attained by projective measurements on  two mutually unbiased bases \cite{mubs}. For example, one can pick the computational basis
\begin{align}\label{computational}
\set B_C =  \{  |j\>  \, ,  j  =   0,\dots,  d-1  \}
\end{align} and the Fourier basis
\begin{align}\label{fourier}
\set B_F   =\{  |e_k\>  \, , k  =  0, \dots,   d-1\} \, ,
\end{align}
defined by
\begin{equation}
|e_k\>    =  \frac 1 {\sqrt d}  \,  \sum_{j=0}^{d-1}   \,  e^{\frac{ 2\pi  i  jk}d}  \,  |j\> \, .
\end{equation}
For the corresponding  pair of projective POVMs, Haapasalo   \cite{haapasalo}  has shown that
\begin{align}
\label{eq:defmu}
\Rm({\st P}, {\st Q}) = \frac{1}{2}\left(1-\frac{1}{\sqrt d }\right).
\end{align}

In the case of testers,  we conjecture  that the maximum of the measurement-induced incompatibility is reached when the canonical POVMs are  measurements on two mutually unbiased bases for the input-output Hilbert space $\spc H  =  \spc H_1\otimes \spc H_0$.    For example, we can choose the computational and Fourier bases  in Eqs.~(\ref{computational}) and (\ref{fourier}), with $d  :=  d_1  d_0$, and define the testers $\st A$ and $\st B$ with
\begin{align}
\label{eq:defPQpair}
A_j=\frac{1}{d_0}  |j\>\<  j|  \qquad {\rm and} \qquad  B_k=\frac{1}{d_0}     |e_k\>\< e_k|\,,
\end{align}
and $j,k$ ranging from $0$ to $d-1$.
Both testers  $\st A$ and $\st B$ have the normalization state $\rho  =  I_0/d_0$.  Their canonical POVMs $\st P$ and $\st Q$ are given by
\begin{equation}
P_j  =  |j\>\<  j|  \qquad {\rm and}  \qquad Q_k  =   |e_k\>\< e_k|,
\end{equation}
respectively.      Combining Eqs.~(\ref{measureinc}) and (\ref{eq:defmu}), we obtain  the bound
\begin{equation}
\Rt({\st A},  {\st B})   \le  \Rm(   {\st P},  {\st Q})  =   \frac{1}{2}\left(1-\frac{1}{\sqrt {d_0 d_1}}\right).
\end{equation}
We conjecture that the r.h.s.~is the maximum amount of measurement-induced incompatibility that can be observed for a process with $d_0$-dimensional input space  and $d_1$-dimensional output space.

\section{Proof of proposition \ref{prop:purenormcons}. }\label{app:purenormcons}

The proof consists of three steps: \\

\noindent {\em Step 1.}  We show that,   for pure normalization states, the testers $\st A$ and $\st B$ coincide with their canonical POVMs $\st P$ and $\st Q$, respectively.    Denote the normalization state by $\rho  = |\psi\>\<\psi|$.
By definition, one has
\begin{equation}
\sum_j A_j=\sum_k B_k=I_1\otimes |\psi\>\<\psi|.
\end{equation}
This condition implies that $A_j$ and $B_k$ have the product form
$A_j=a_j\otimes  |\psi\>\<\psi|$ and $B_k=b_k  \otimes  |\psi\>\<\psi| $, where ${\st a}  = \{a_j\}$ and ${\st b}  = \{b_k\}$ are   POVMs.     Note that the canonical POVMs act on the Hilbert space $\spc H_0\otimes \Supp (  |\psi\>\<\psi|)   \simeq  \spc H_0$ and   satisfy  $\st P   =  \st A   \simeq  \st a$ and $\st Q  =  \st B  \simeq \st b$.
\medskip

\noindent{\em Step 2.}  We show that, for the evaluation of the robustness of incompatibility, it is enough to restrict the attention to noise testers ${\st N}^{(\st A)}$ and   ${\st N}^{(\st B)}$  with the same normalization state  $\rho  = |\psi\>\<\psi|$. Indeed, assume that the mixed testers
\begin{equation}
(1-\lambda)\,   {\st A }  +  \lambda\,  {\st N}^{(\st A)}  \quad {\rm and}  \quad      (1-\lambda)\,   {\st B }  +  \lambda\,  {\st N}^{(\st B)}
\end{equation}
are compatible.   Compatibility means that there exists a tester $\st C$ such that
\begin{align}
\nonumber \sum_k C_{jk} &=(1-\lambda) \, A_j+\lambda  \,  N^{(\st A)}_j,\\
\label{sandwich}
\sum_j C_{jk} &=(1-\lambda) \, B_k+\lambda \,   N^{(\st B)}_k \, ,
\end{align}
for all $j$ and $k$.
Denoting by $\widetilde \rho$, $\widetilde \sigma$, and $\omega$ the normalization states of ${\st N}^{(\st A)}$, ${\st N}^{(\st B)}$, and $\st C$, respectively, the above compatibility relations imply
\begin{align}  \omega  &  =   (1-\lambda)  \,    |\psi\>\<\psi|  +   \lambda\,  \widetilde \rho   \nonumber\\
&    =   (1-\lambda)  \,    |\psi\>\<\psi|  +   \lambda\,  \widetilde \sigma   \, ,
\end{align}
and, of course,
\begin{equation}
\widetilde \rho  = \widetilde \sigma.
\end{equation}
Now,  define the testers  $ \widetilde {\st   N}^{(\st A)}$,  $ \widetilde {\st   N}^{(\st B)}$, and $\widetilde {\st C}$ with operators
\begin{align}
\widetilde   N_j^{(\st A)}  &:= \frac{(I_1\otimes  |\psi\>\<\psi|  )   \,  N_j^{(\st A)} \,  (I_1\otimes  |\psi\>\<\psi| ) }{ \<\psi|  \widetilde \rho| \psi\> } , \nonumber\\ \nonumber \\
\widetilde   N_k^{(\st B)} & : = \frac{(I_1\otimes  |\psi\>\<\psi|  )   \,  N_k^{(\st B)} \,  (I_1\otimes  |\psi\>\<\psi| ) }{ \<\psi|  \widetilde \rho| \psi\> }  ,
\\ \nonumber \\
\widetilde C_{jk} &: =   \frac{(I_1\otimes  |\psi\>\<\psi|  )   \,  C_{jk}\,  (I_1\otimes  |\psi\>\<\psi| ) }{ \<\psi|  \omega| \psi\> }    .\nonumber
\end{align}
By definition, all these testers have the same normalization state, equal to $\rho  =  |\psi\>\<\psi|$.
Moreover, pinching both sides of Eq.~(\ref{sandwich}) with the projector  $(I_1\otimes |\psi\>\<\psi|)$ we  obtain the relation
\begin{align}
\nonumber \sum_k \widetilde C_{jk} &= (1-\widetilde \lambda) \, A_j+  \widetilde \lambda \, \widetilde N^{(\st A)}_j,\\
\label{compatibility2}
\sum_j \widetilde  C_{jk} &= (1-\widetilde \lambda) \, B_k+\widetilde \lambda \, \widetilde N^{(\st B)}_k ,
\end{align}
valid for every $j$ and $k$ with
\begin{align}
\widetilde  \lambda  &=     \frac{ \<\psi  | \widetilde\rho |\psi\>}{  \<  \psi| \omega |\psi\> }\, \lambda .
\end{align}
Using  the relation
\begin{align*}
\frac{ \<\psi  | \widetilde\rho |\psi\>}{  \<  \psi| \omega |\psi\> }    =       \frac{ \<\psi  | \widetilde\rho |\psi\>}{ (1-\lambda)  +  \lambda  \, \<\psi| \widetilde \rho  |\psi\> }
     \le 1,
\end{align*}
 we then obtain
 \begin{equation}
\widetilde \lambda \le \lambda.
\end{equation}
In conclusion, the search for the minimum  $\lambda$ can be restricted without loss of generality to noise testers ${\st N}^{(\st A)}$ and ${\st N}^{(\st B)}$  of the form
\begin{align}
N_j^{(\st A)}    &=   n_j^{(\st A)}  \otimes  |\psi\>\<\psi| \nonumber \\
N_j^{(\st B)}    &=   n_j^{(\st B)}  \otimes  |\psi\>\<\psi|\,,
\end{align}
 where    ${\st n}^{(\st A)}  =  \left\{n_j^{(\st A)} \right \}$ and ${\st n}^{(\st B)}  =  \left\{n_j^{(\st B)} \right \}$ are suitable POVMs.
 \medskip

 \noindent{\em Step 3.}   Note that the $\lambda$-compatibility conditions (\ref{compatibility2}) are  equivalent to the $\lambda$-compatibility of the canonical POVMs. Minimizing over $\lambda$, we then obtain the lower bound  $\Rt({\st A}, {\st B})\ge \Rm({\st P},{\st Q})$. Combining this lower bound with  the  upper bound of Eq.    (\ref{measureinc}), we obtain the equality $\Rt({\st A}, {\st B}) =  \Rm ({\st P},{\st Q})$.
\qed

\section{Proof of proposition \ref{prop:comparable}}\label{app:comparable}
\begin{proof}
Without loss of generality, let us  assume  $A_1\leq B_1$.     Let $ \rho $ and $\sigma$ be the normalization states    of the testers $\st A$ and $\st B$, respectively, and suppose that the relation
\begin{equation}
(1-\lambda)  \,  \rho   + \lambda\,  \widetilde \rho  =    (1-\lambda)  \,  \sigma  + \lambda\,  \widetilde \sigma,
\end{equation}
holds for suitable density operators $\widetilde \rho$ and $\widetilde \sigma$. Then
define the  testers ${\st N}^{(\st A)},   {\st N}^{(\st B)}$ and $\st C$  with operators
\begin{align}
 N_1^{(\st A)}   & =  0,  &
 N_2^{(\st A)}  &=  I_1\otimes \widetilde \rho,   \nonumber\\
 N_1^{(\st B)}  &=  0,    &
  N_2^{(\st B)}   & =  I_1\otimes \widetilde \sigma     \, ,
\end{align}
and
\begin{align}
C_{11}&=(1-\lambda)A_1,\nonumber\\
 C_{12}&=0,\nonumber \\
C_{21}&=(1-\lambda)(B_1-A_1),\nonumber \\
 C_{22}&=(1-\lambda)(I\otimes\sigma-B_1)+\lambda I\otimes\widetilde{\sigma}  .
\end{align}
With the above definitions, one has
\begin{align}
\sum_{k} C_{jk}&= (1-\lambda) \, A_j+\lambda  \,N^{(\st A)}_j\nonumber\,,\\
\sum_{j} C_{jk}&= (1-\lambda) \,B_k+\lambda   \,  N^{(\st B)}_k\,.
\end{align}
Hence, $\st A$ and $\st B$ are $\lambda$-compatible.  Minimizing over $\lambda$, we obtain the upper bound $\Rt({\st A},  {\st B})  \le \Rs(\rho,\sigma)$.   Combining this bound with the lower bound of Eq.~(\ref{lowerb}), we then have the equality   $\Rt({\st A},  {\st B})  = \Rs(\rho,\sigma)$.
\end{proof}

\section{Proof of Proposition \ref{prop:achievinM}}
\label{app:achievinM}

We will prove the Proposition~\ref{prop:achievinM} by demonstrating a particular choice of testers  ${\st N}^{(\st A)},   {\st N}^{(\st B)}$ and a joint tester $\st C$ satisfying all the requirements of the robustness of incompatibility for $\lambda=\Rt(\st A,\st B)=\Rs(P_{-\theta/2},P_{\theta/2})$ as specified in Eq.~(\ref{eq:minlforAB}). We further set
\begin{align}
\label{eq:defABtildaq}
N_1^{(\st A)}&=\Big[\frac{1-\delta}{2} P_{(\varphi+\pi)/2}+\frac{1+\delta}{2} P_{(\varphi-\pi)/2}\Big] \otimes P_{\pi/2},   \nonumber \\
N_2^{(\st A)}&=I \otimes \widetilde \rho - N_1^{(\st A)},   \nonumber \\
N_1^{(\st B)}&=\Big[\frac{1-\delta}{2} P_{-(\varphi+\pi)/2}+\frac{1+\delta}{2} P_{-(\varphi-\pi)/2}\Big] \otimes P_{-\pi/2}, \nonumber \\
N_2^{(\st B)}&=I \otimes \widetilde\sigma - N_1^{(\st B)},
\end{align}
where
\begin{align}
\label{eq:defdelta}
\delta&=-\,\frac{\sin \varphi}{2}\,\frac{1- \sin \frac{\theta}{2}}{ \sin \frac{\theta}{2}}\,.
\end{align}
Let us stress that the associated states
$\widetilde \rho=P_{\pi/2}$ and $\widetilde\sigma=P_{-\pi/2}$
are orthogonal as it is required in order to saturate the bound (\ref{eq:minlforAB}).
By definition we can express the mixed state $\omega$ as
\begin{align}
\label{eq:sformomega}
\omega&=\frac{1}{2}\left[(1-\lambda)(\rho+\sigma)+\lambda (\widetilde{\rho}+\widetilde{\sigma})\right] \nonumber \\
&=\frac{1-\lambda}{2}(P_{-\theta/2}+P_{\theta/2})+\frac{\lambda}{2} I,
\end{align}
which will be convenient in subsequent calculations.
Using $\st{\bar A}$ and $\st{\bar B}$ for the mixed testers
we define the joint tester $\st C=\{C_{11},C_{12},C_{21},C_{22}\}$ as follows
\begin{align}
C_{11}&= C,  & C_{12}&=\bar{A}_1-C,  \\
C_{21}&= \bar{B}_1-C,  & C_{22}&=I\otimes \omega +C -\bar{A}_1-\bar{B}_1, \nonumber
\end{align}
where
\begin{align}
\label{eq:defG}
C=&\frac{1-\lambda}{2} (\cos^2 \frac{\varphi}{2}+\sin^2 \frac{\varphi}{2}\sin \frac{\theta}{2})\Big[\ket{v_1}\bra{v_1}+\ket{v_2}\bra{v_2}\Big] \nonumber \\
&+\frac{1-\lambda}{2} \,\cos \frac{\varphi}{2}\cos \frac{\theta}{2}\,\Big[\ket{v_1}\bra{v_2}+\ket{v_2}\bra{v_1}\Big]
\end{align}
and
\begin{align}
\ket{v_1}&=\ket{\frac{\varphi}{2}}\ket{\frac{\pi}{2}}, \quad \quad \ket{v_2}=\ket{-\frac{\varphi}{2}}\ket{-\frac{\pi}{2}},  \nonumber\\
\ket{\beta}&=\cos \frac{\beta}{2}\ket{0} + \sin \frac{\beta}{2}\ket{1}.
\end{align}

To demonstrate $\lambda$-compatibility of $\st A$, $\st B$ it suffices to show (see  Proposition \ref{prop:SDP}) that
\begin{align}
\label{eq:ineqtoshow}
0\leq C &\leq \bar{A}_1, \bar{B}_1,\\
\bar{A}_1+\bar{B}_1 &\leq C+I\otimes\omega
\label{eq:ineqDG}
\end{align}
holds.
Since $\<v_1\ket{v_2}=0$, the nonzero eigenvalues of $C$ are the same as for the matrix
\begin{align}
\label{eq:matg}
\frac{1-\lambda}{2}
\left(\begin{array}{cc}
a & b \\
b & a
\end{array}
\right),
\end{align}
where $a=\cos^2 \frac{\varphi}{2}+\sin^2 \frac{\varphi}{2}\sin \frac{\theta}{2}$, $b= \cos \frac{\varphi}{2}\cos \frac{\theta}{2}$.
After some algebra the requirement of non-negativity of the eigenvalues leads to the definition of the region $\mathsf M$. Thus, in region $\mathsf M$ we proved $C\geq 0$.

Let us define
\begin{align}
\label{eq:defD}
&D\equiv \bar{A}_1+\bar{B}_1-I\otimes\omega, \nonumber\\
&Q\equiv P_{\varphi/2}\otimes P_{\pi/2}+P_{-\varphi/2}\otimes P_{-\pi/2}, \nonumber\\
&Q^\perp=I-Q, \nonumber\\
&S\equiv \sigma_X \otimes \sigma_Z,
\end{align}
where $\sigma_X,\sigma_Z$ are the Pauli matrices.
Then Eq.~(\ref{eq:ineqDG}) can be rewritten as $D\leq C$.
In the Appendix \ref{app:diagD} we prove
\begin{align}
\label{eq:D}
D&=QDQ+Q^\perp D Q^\perp=C-S C S,
\end{align}
which implies $D\leq C$, because $C-D=S C S\geq 0$ due to preservation of eigenvalues of $C\geq 0$ by unitary rotation $S$.
Thus, we proved Eq.~(\ref{eq:ineqDG}).

Next, we show that due to the symmetry of the problem
\begin{align}
\label{eq:GAGB}
C &\leq \bar{A}_1 \Leftrightarrow C \leq \bar{B}_1
\end{align}
For this purpose we define hermitian and unitary operator $T\equiv \sigma_Z \otimes \sigma_Z$, for which $T^2=I$.
It is easy to verify by direct calculation from Eqs. (\ref{eq:defABq}), (\ref{eq:defABtildaq}), (\ref{eq:defG}) that
\begin{align}
\label{eq:GAGBproof1}
TCT=C \quad  T \bar{A}_1 T= \bar{B}_1.
\end{align}
Since conjugation with $T$ is reversible and preserves eigenvalues we get
\begin{align}
\label{eq:GAGBproof2}
C\leq\bar{A}_1\;\Leftrightarrow\; TCT\leq T \bar{A}_1 T \;\Leftrightarrow \;C \leq \bar{B}_1,
\end{align}
where we used Eq.~(\ref{eq:GAGBproof1}).

In the following we prove $C\leq\bar{A}_1$ by demonstrating positivity of the matrix of $\bar{A}_1-C$ in the basis $\{\ket{v_1},\ket{v_2},\ket{v_3},\ket{v_4}\}$, where
\begin{align}
\ket{v_3}&=S\ket{v_2}=\ket{\frac{\varphi}{2}-\pi}\ket{\frac{\pi}{2}}, \nonumber \\
\ket{v_4}&=S\ket{v_1}=\ket{\pi-\frac{\varphi}{2}}\ket{-\frac{\pi}{2}}.
\end{align}
Direct calculation of the matrix elements yields
\begin{align}
\label{eq:A1Gmat}
\bar{A}_1-C&=
\left(\begin{array}{cccc}
0 & 0 & 0 & 0 \\
0 & x & y & 0 \\
0 & y & z & 0 \\
0 & 0 & 0 & 0
\end{array}
\right),
\end{align}
where
\begin{align}
\label{eq:defxyz}
x&=\frac{1}{2}\left( 1- \frac{\cos^2 \frac{\varphi}{2}+\sin^2 \frac{\varphi}{2}\sin \frac{\theta}{2}}{1+\sin \frac{\theta}{2}}\right), \nonumber \\
y&= \frac{\sin \frac{\varphi}{2}\cos \frac{\theta}{2}}{2(1+\sin \frac{\theta}{2})}, \\
z&= \frac{1}{2}\left( 1- \frac{\sin^2 \frac{\varphi}{2}\sin \frac{\theta}{2}}{1+\sin \frac{\theta}{2}}\right)  . \nonumber
\end{align}
Thus, it suffice to examine eigenvalues of matrix
\begin{align}
\label{eq:A1Gmsmall}
W=
\left(\begin{array}{cc}
 x & y \\
 y & z
\end{array}
\right),
\end{align}
which can be analytically shown to be non-negative for all $\theta,\varphi\in[0,\pi]$.
In conclusion we proved validity of Eqs.~(\ref{eq:ineqtoshow}), (\ref{eq:ineqDG}) for all $(\theta,\varphi)\in\mathsf M$  and thus demonstrated the existence of the joint tester $\st C$ needed for proving $\Rt(\st A,\st B)=\Rs(P_{-\theta/2},P_{\theta/2})$ claimed in the proposition.

\section{Diagonal form of the operator $D$}
\label{app:diagD}
Let us first explicitly write operator $D$
\begin{align}
\label{eq:Dexplicit}
D=&\frac{1-\lambda}{2}\Big[(P_{-\varphi/2}-P_{\pi-\varphi/2}) \otimes P_{-\theta/2} \nonumber \\
& \quad \quad\;\; +(P_{\varphi/2}-P_{\varphi/2-\pi})\otimes P_{\theta/2}\Big] \nonumber \\
&+\frac{\lambda\delta}{2}\Big[ (P_{(\pi-\varphi)/2}-P_{-(\pi+\varphi)/2}) \otimes P_{-\pi/2} \nonumber \\
& \quad \quad\;\; +(P_{(\varphi-\pi)/2}-P_{(\pi+\varphi)/2}) \otimes P_{\pi/2} \Big]\, ,
\end{align}
where we used (\ref{eq:sformomega}), (\ref{eq:defD}). It can be written more compactly as
\begin{align}
\label{eq:Dshort}
D=\frac{1-\lambda}{2}(H+THT)+ \frac{\lambda\delta}{2}(K+TKT),
\end{align}
where
\begin{align}
\label{eq:defHK}
H&=(P_{-\varphi/2}-P_{\pi-\varphi/2}) \otimes P_{-\theta/2}, \nonumber \\
K&=(P_{(\pi-\varphi)/2}-P_{-(\pi+\varphi)/2}) \otimes P_{-\pi/2},
\end{align}
and $T\equiv \sigma_Z \otimes \sigma_Z$ is a tensor product of Pauli matrices.

Our aim is to show that operator $D$ does not mix subspaces defined by projectors $Q,Q^\perp$, i.e.
\begin{align}
\label{eg:QDQp}
QDQ^\perp=Q^\perp DQ=0.
\end{align}
Thanks to hermicity of operator $D$ it suffices to show $QDQ^\perp=0$.
We observe that $TQT=Q$ and consequently $[Q,T]=0$. Similarly, $TQ^\perp T=Q^\perp$ implies $[Q^\perp,T]=0$. This means it is crucial to calculate operators $QHQ^\perp$, $QKQ^\perp$ and the remaining terms of $QDQ^\perp$ can be obtained by conjugation with $T$. For such calculation the following formula is useful
\begin{align}
\label{eg:Pproducts}
P_{\alpha} P_{\beta} P_{\gamma} = \ket{\alpha}\bra{\gamma}\;\cos\frac{\alpha-\beta}{2}\cos\frac{\beta-\gamma}{2}.
\end{align}
After a longer, but straightforward calculation one obtains
\begin{align}
\label{eq:HKresult1}
QHQ^\perp=&\frac{\sin \frac{\varphi}{2} \cos\frac{\theta}{2} }{2}  \Big[\ket{-\frac{\varphi}{2}}\bra{\frac{\varphi}{2}-\pi}\otimes\ket{-\frac{\pi}{2}}\bra{\frac{\pi}{2}} \nonumber \\
 &\quad\quad\quad\quad\quad\;\; - \ket{\frac{\varphi}{2}}\bra{\pi-\frac{\varphi}{2}}\otimes\ket{\frac{\pi}{2}}\bra{-\frac{\pi}{2}} \Big] \nonumber \\
 &+\sin \varphi\frac{1-\sin\frac{\theta}{2}}{2}  \ket{\frac{\varphi}{2}}\bra{\frac{\varphi}{2}-\pi}\otimes P_{\pi/2},    \nonumber \\
QKQ^\perp=&\ket{-\frac{\varphi}{2}}\bra{\pi-\frac{\varphi}{2}} \otimes P_{-\pi/2}.
\end{align}
Thanks to Eq.~(\ref{eq:HKresult1}) it is easy to evaluate
\begin{align}
\label{eq:HKresult2}
&QHQ^\perp+TQHQ^\perp T=\sin \varphi \; \frac{1-\sin\frac{\theta}{2}}{2}\; \times     \nonumber \\
&\quad\;\;\times \Big[ \ket{\frac{\varphi}{2}}\bra{\frac{\varphi}{2}-\pi}\otimes P_{\pi/2} + \ket{-\frac{\varphi}{2}}\bra{\pi-\frac{\varphi}{2}}\otimes P_{-\pi/2}\Big],  \nonumber \\
&QKQ^\perp+T QKQ^\perp T=  \nonumber \\
&\quad\;\;= \Big[ \ket{\frac{\varphi}{2}}\bra{\frac{\varphi}{2}-\pi}\otimes P_{\pi/2} + \ket{-\frac{\varphi}{2}}\bra{\pi-\frac{\varphi}{2}}\otimes P_{-\pi/2}\Big],
\end{align}
where the terms with $\cos\frac{\theta}{2}$ effectively disappeared due to conjugation.
Finally, using Eqs.~(\ref{eq:Dshort}), (\ref{eq:HKresult2}) and definitions (\ref{eq:minlforAB}), (\ref{eq:defdelta}) we get $QDQ^\perp=0$, because
\begin{align}
\frac{1-\lambda}{2}\sin \varphi \; \frac{1-\sin\frac{\theta}{2}}{2}\; +\frac{\lambda\delta}{2}=0.
\end{align}
This allows us to write
\begin{align}
\label{eq:fineqHK}
D=(Q+Q^\perp)D(Q+Q^\perp)=QDQ+ Q^\perp D Q^\perp.
\end{align}
Let us calculate $QDQ$.
Direct calculation using (\ref{eg:Pproducts}) shows that
\begin{align}
\label{eq:QKQ}
QKQ&=\Big[P_{-\varphi/2}\left(P_{(\pi-\varphi)/2}-P_{-(\pi+\varphi)/2}\right)P_{-\varphi/2}\Big]\otimes P_{-\pi/2} \nonumber \\
   &= \Big[\frac{1}{2} P_{-\varphi/2} - \frac{1}{2} P_{-\varphi/2} \Big]\otimes P_{-\pi/2}=0
\end{align}
and
\begin{align}
\label{eq:QHQ}
QHQ=&\frac{1-\sin\frac{\theta}{2}}{2}\cos \varphi \ket{v_1}\bra{v_1} +
\frac{1+\sin\frac{\theta}{2}}{2} \ket{v_2}\bra{v_2} \nonumber \\
&+\frac{1}{2}\cos\frac{\theta}{2}\cos\frac{\varphi}{2} \left(\ket{v_1}\bra{v_2}+\ket{v_2}\bra{v_1}\right).
\end{align}
Thanks to Eqs.~(\ref{eq:Dshort}), (\ref{eq:QKQ}), (\ref{eq:QHQ}) and the fact that $T\ket{v_1}=\ket{v_2}$ we obtain
\begin{align}
\label{eq:QDQfin}
QDQ=\frac{1-\lambda}{2}(QHQ+TQHQT)=C\, .
\end{align}
The last thing we have to show is $Q^\perp D Q^\perp=-SCS$.
Let us note the following identities
\begin{align}
\label{eq:QpDQpstart}
SHS&=-THT, \qquad ST=-TS, \nonumber \\
SKS&=-TKT, \qquad SQS=Q^\perp  ,
\end{align}
which from Eqs.~(\ref{eq:Dshort}), (\ref{eq:QDQfin}) imply $SDS=-D$ and
\begin{align}
\label{eq:QpDQp}
&Q^\perp D Q^\perp=SQSDSQS=-SQDQS=-SCS .
\end{align}
Combining equations (\ref{eq:fineqHK}), (\ref{eq:QDQfin}), (\ref{eq:QpDQp}) we obtain Eq.~(\ref{eq:D}) we wanted to prove.\qed

\section{Proof of Proposition \ref{prop:qABpovm}}
\label{app:proofqABpovm}

Since $E_2=I-E_1$ and $F_2=I-F_1$, we can parametrize both POVMs by spectral decompositions of the effects $E_1$, $F_1$
\begin{align}
\label{eq:defEFPOVM}
E_1&=e_1 \ket{u_1}\bra{u_1}+e_2 \ket{u_2}\bra{u_2}\nonumber \\
F_1&=f_1 \ket{w_1}\bra{w_1}+f_2 \ket{w_2}\bra{w_2},
\end{align}
where $e_i,f_j\in[0,1]$ and $\{\ket{u_1},\ket{u_2}\}$, $\{\ket{w_1},\ket{w_2}\}$ are two orthonormal qubit bases.
Effects $E_1$, $F_1$ as well as the corresponding POVMs can be convexly decomposed into four projective measurements (extremal POVMs),
\begin{align}
\label{eq:defEFconv}
E_1&=\sum_{a=1}^4 c_a E^a_1, \quad \quad F_1=\sum_{b=1}^4 d_b F^b_1,
\end{align}
where
\begin{align}
\label{eq:defEiFj}
E^1_1&=0, \quad &F^1_1&=0, \nonumber \\
E^2_1&=\ket{u_1}\bra{u_1}, \quad &F^2_1&=\ket{w_1}\bra{w_1}, \nonumber \\
E^3_1&=\ket{u_2}\bra{u_2}, \quad &F^3_1&=\ket{w_2}\bra{w_2}, \\
E^4_1&=I, \quad &F^4_1&=I. \nonumber
\end{align}
The decomposition in Eq.~(\ref{eq:defEFconv}) is unique and such that
\begin{align}
\label{eq:sumconvcoef}
\sum_{a=1}^4c_a&=1, \quad &\sum_{b=1}^4d_b&=1.
\end{align}
The two outcome POVMs defined by effects $E^1_1$, $E^4_1$, $F^1_1$, $F^4_1$ are trivial, i.e.~their outcomes can be generated without actually measuring the quantum state.
The first pair and the second pair of POVMs defined by effects $E^2_1$, $E^3_1$, $F^2_1$, $F^3_1$ are related by relabeling of outcomes (e.g.~$E^2_1=E^3_2$, $E^2_2=E^3_1$).

We define $1$-testers
\begin{align}
{\st A}^a&=\{A^a_1,A^a_2\},&\st B^b&=\{B^b_1, B^b_2\}, \nonumber \\
A^a_k&=E^a_k\otimes P_{-\theta/2},  &B^b_k&=F^b_k\otimes P_{\theta/2} ,
\end{align}
where $k=1,2$ and $a,b=1,2,3,4$. Showing that all pairs of testers $\st A^a$ and $\st B^b$ (for all $a$, $b$) are $\lambda$-compatible will be later used to show compatibility of 1-testers $\st A$ and $\st B$.

Firstly, the outcomes of trivial POVMs can be generated without measuring the state, and so it is clear that those $1$-testers defined above that contain trivial POVM will be $\lambda$-compatible with $\lambda$ obeying Proposition \ref{prop:normalization} with any other product $1$-tester $\st B=\{F_1\otimes \sigma,F_2\otimes \sigma\}$.
For example, for $A^4_1=I\otimes \rho$
it would suffices to choose the joint and the admixed testers as
$\widetilde{A}^4_1=I\otimes\widetilde\rho$, $\widetilde{B}_1=F_1\otimes \widetilde\sigma$ and
\begin{align}
C_{11}&=(1-\lambda)B_1  +\lambda \widetilde{B}_1,        &C_{12}&=(1-\lambda)B_2  +\lambda \widetilde{B}_2, \nonumber \\
C_{21}&=0,          &C_{22}&=0 .
\end{align}
Clearly,
$C_{11}+C_{21}=\bar{B}_1$, $C_{12}+C_{22}=\bar{B}_2$ and
\begin{align}
C_{21}+C_{22}&=0=\bar{A}^4_2,  \nonumber \\
C_{11}+C_{12}&=(1-\lambda)I\otimes\sigma + \lambda I\otimes \widetilde{\sigma}  \nonumber \\
   &=(1-\lambda)I\otimes\rho + \lambda I\otimes \widetilde{\rho}= \bar{A}^4_1.
\end{align}
Thus, for $\theta \geq 2 \arcsin (1/3)\approx 0.6797$ $1$-testers $\st A^a$, $\st B^b$ are $\lambda$-compatible, because either one of them contains trivial POVM or the pair is unitarily equivalent to $1$-testers in Proposition \ref{prop:achievinM}.

Now it suffices to show that $\lambda$-compatibility of $1$-testers $\st A^a$, $\st B^b$ for all $a,b$ implies $\lambda$-compatibility of $1$-testers $\st A$, $\st B$ from the proposition \ref{prop:qABpovm}.
This can be done as follows.
The fact that for $\theta \geq 2 \arcsin (1/3)$ $1$-testers $\tA^a$, $\tB^b$ are $\lambda$-compatible can be expressed using Proposition \ref{prop:SDP} by existence of operators $C^{ab}$ satisfying
\begin{align}
\label{eq:ineqGABab}
0\leq C^{ab} &\leq \bar{A}^{a}_1, \bar{B}^b_1,\\
\label{eq:ineqGABab2}
\bar{A}^{a}_1+\bar{B}^b_1 &\leq C^{ab}+I\otimes\omega.
\end{align}
Let us note that $\forall a,b$ $A^{a}_1+A^{a}_2=I\otimes P_{-\theta/2}$, $B^{b}_1+B^{b}_2=I\otimes P_{\theta/2}$ and since $\lambda$ is given by the lower bound (\ref{eq:minlforAB}) also
$\forall a,b$ $\widetilde A^{a}_1+\widetilde A^{a}_2=I\otimes P_{\pi/2}$, $B^{b}_1+B^{b}_2=I\otimes P_{-\pi/2}$ and as a consequence the normalization of the joint tester $\omega$ is the same $\forall a,b$.

To prove $\lambda$-compatibility of $\st A$, $\st B$ we define the admixed $1$-testers $\st{\widetilde A}$, $\st{\widetilde B}$ and the joint $1$-tester $\st{C}$
\begin{align}
\widetilde A_1&=\sum_{a=1}^4 c_a \widetilde A^{a}_1, &\widetilde B_1&=\sum_{b=1}^4 d_b \widetilde B^{b}_1, \nonumber\\
C&=\sum_{a,b=1}^4 c_a d_b C^{ab}.
\end{align}
Let us remind that $A_1=\sum_{a=1}^4 c_a A^{a}_1$, $B_1=\sum_{b=1}^4 d_b B^{b}_1$, which implies
\begin{align}
\bar{A}_1=\sum_{a=1}^4 c_a \bar{A}^{a}_1, \quad  &\bar{B}_1=\sum_{b=1}^4 d_b \bar{B}^{b}_1.
\end{align}
Due to $C^{ab}\geq 0$, $c_a,d_b\geq0$ we conclude $C\geq 0$, because $C$ is a nonnegative sum of positive-semidefinite operators.
We also easily get
\begin{align}
C&=\sum_{a,b=1}^4 c_a d_b C^{ab}\leq \sum_{a,b=1}^4 c_a d_b \bar{A}^{a}_1= \bar{A}_1, \nonumber \\
C&=\sum_{a,b=1}^4 c_a d_b C^{ab}\leq \sum_{a,b=1}^4 c_a d_b \bar{B}^{b}_1= \bar{B}_1,
\end{align}
where we used (\ref{eq:ineqGABab}) and (\ref{eq:sumconvcoef}).

Finally, we use Eq.~(\ref{eq:ineqGABab}) to write
\begin{align}
\bar{A}_1+\bar{B}_1&=\sum_{a,b=1}^4 c_a d_b ( \bar{A}^{a}_1+\bar{B}^{b}_1 )\leq I\otimes \omega + C ,
\end{align}
which concludes the proof.\qed

\section{SDP for $\lambda$-compatibility}\label{app:SDP}

Proposition \ref{prop:SDP} can be used to construct SDP for solving $\lambda$-compatibility of two two-outcome testers $\st A=\{A_1,A_2\}$ and $\st B=\{B_1,B_2\}$ such that $A_1+A_2=I\otimes\rho$ and $B_1+B_2=I\otimes\sigma$. According to Definition \ref{lambdacomp} in $\lambda$-compatibility we search for the smallest $\lambda$ such that the testers $(1-\lambda)\st A+\lambda\st{\widetilde A}$ and $(1-\lambda)\st B+\lambda\st{\widetilde B}$ are compatible for some $\st{\widetilde A}$ and $\st{\widetilde B}$.
First of all, the necessary condition $I \otimes \bar \rho = I\otimes \bar \sigma$ from Eq.~(\ref{eq:ncnorm}) for the normalizations $\bar \rho=(1-\lambda)\rho+\lambda \widetilde\rho$ and $\bar \sigma = (1-\lambda)\sigma + \lambda\widetilde\sigma$ needs to be satisfied. Thus,
in addition to search over operators $C$, we have to expand the search also over the mixed-in elements $\widetilde A_1$, $\widetilde B_1$  and their normalizations $\widetilde{\rho}$, $\widetilde{\sigma}$. Defining
\begin{align}
\bar A_i &=(1-\lambda)A_i+\lambda \widetilde A_i,\nonumber\\
\bar B_j &=(1-\lambda)A_j+\lambda \widetilde B_j,\\
\omega &=(1-\lambda)\rho+\lambda\widetilde\rho,\nonumber\\
&= (1-\lambda)\sigma+\lambda\widetilde\sigma,\nonumber
\end{align}
the problem can be recast as the following bilinear SDP
\begin{align}\begin{array}{ll}
\text{Find}    &  \lambda_0:=\inf\lambda\\
\text{subject to }  &0 \leq \widetilde A_1, \widetilde B_1, \widetilde\rho, \widetilde\sigma, C\\
&C \leq \bar A_1,\bar B_1\\
&\bar A_1+\bar B_1\leq C+I\otimes\omega,\\
&(1-\lambda)(\rho-\sigma)=\lambda(\widetilde\sigma-\widetilde\rho),\\
&\widetilde A_1\leq I\otimes\widetilde\rho,\widetilde B_1\leq I\otimes\widetilde\sigma,\\
&\tr[\widetilde\rho]=\tr[\widetilde\sigma]=1,
\end{array}
\end{align}
where the last condition comes from the common normalization to $\omega$. We can linearize the program by rescaling relevant operators by $1/\lambda$. Using the definition of $\omega$ and setting $\mu=(1-\lambda)/\lambda$, $H=\frac{1}{\lambda}C$ the SDP program can be equivalently stated as
\begin{align}
\begin{array}{ll}
\text{Find}  &\mu_0:=\sup\mu\\
\text{subject to}
&0\leq\widetilde A_1, \widetilde B_1, \widetilde\rho, \widetilde\sigma, H,\\
&H\leq\mu A_1+\widetilde A_1,\\
&H\leq\mu B_1+\widetilde B_1,\\
&\mu(A_1+B_1-I\otimes\sigma)+\widetilde A_1+\widetilde B_1\leq H+I\otimes\widetilde\sigma,\\
&\mu(\rho-\sigma)=\widetilde\sigma-\widetilde\rho,\\
&\widetilde A_1\leq I\otimes\widetilde\rho,\widetilde B_1\leq I\otimes\widetilde\sigma,\\
&\tr[\widetilde\rho]=\tr[\widetilde\sigma]=1,
\end{array}
\end{align}
where the unknown objects are $\mu, H, \widetilde A_1, \widetilde B_1, \widetilde\rho, \widetilde\sigma$.
Then the minimal $\lambda$ is determined as
\begin{equation}
\lambda_0=\frac{1}{1+\mu_0}.
\end{equation}

\end{document}